\documentclass[11pt,letter]{article}

\usepackage{amsmath}
\usepackage{amssymb}
\usepackage{amsthm}
\usepackage{fancyhdr}
\usepackage{makeidx}
\usepackage{url}
\usepackage{graphicx}
\usepackage{verbatim} 

\theoremstyle{definition}
\newtheorem{theorem}{Theorem}

\newtheorem{proposition}[theorem]{Proposition}

\theoremstyle{definition}

\newcommand{\mycomment}[1]{}

\newcommand{\E}{\textmd{E}}

\newcommand{\vect}{\textmd{vec}}

\newcommand{\ep}{\epsilon}

\newcommand{\bbE}{\mathbb{E}}

\newcommand{\bbI}{\mathbb{I}}

\newcommand{\bbN}{\mathbb{N}}

\newcommand{\bbP}{\mathbb{P}}

\newcommand{\1}{\mathbf{1}}

\newcommand{\mcE}{\mathcal{E}}
\newcommand{\mcF}{\mathcal{F}}

\newcommand{\mcU}{\mathcal{U}}

\newcommand{\ra}{\rightarrow}
\newcommand{\del}{\partial}
\newcommand{\diag}{\textmd{diag}}

\usepackage[left=1.25in,top=1.0in,right=1.25in,bottom=1.0in]{geometry}
\usepackage{amsfonts}
\usepackage{amssymb}
\usepackage{graphicx}
\usepackage{setspace}
\usepackage[round]{natbib}
\usepackage{amsmath}
\usepackage{amssymb}
\usepackage{amsbsy}
\usepackage{amsthm}
\usepackage{paralist}
\usepackage{bm}
\usepackage{booktabs}

\newcommand{\JJ}{J^*}
\newcommand{\Polya}{P\'{o}lya}
\newcommand{\tablesize}{\footnotesize}
\newcommand{\floor}[1]{\lfloor #1 \rfloor}

\usepackage{mathtools}
\usepackage{cases}
\usepackage{algorithm}
\usepackage{algpseudocode}

\newtheorem{mydefn}{Definition}

\def\E{\qopname\relax o{E}}

\newcommand{\bbeta}{\boldsymbol{\beta}}

\newcommand{\N}{\mbox{N}}
\newcommand{\PG}{\mbox{PG}}

\newcommand{\Ga}{\mbox{Ga}}

\newcommand{\Lo}{\mbox{Lo}}
\newcommand{\KS}{\mbox{KS}}

\title{Bayesian inference for logistic models using \\ \Polya-Gamma latent variables}
\author{Nicholas G. Polson\footnote{ngp@chicagobooth.edu}\\
\textit{\normalsize{University of Chicago}}\\ \\
James G. Scott\footnote{james.scott@mccombs.utexas.edu}\\
Jesse Windle\footnote{jwindle@ices.utexas.edu} \\
\textit{\normalsize{University of Texas at Austin}}}
\date{First Draft: August 2011\\
This Draft: July 2013}

\begin{document}

\maketitle
\begin{abstract}
\noindent We propose a new data-augmentation strategy for fully Bayesian inference in models with binomial likelihoods.  The approach appeals to a new class of \Polya-Gamma distributions, which are constructed in detail.  A variety of examples are presented to show the versatility of the method, including logistic regression, negative binomial regression, nonlinear mixed-effects models, and spatial models for count data.  In each case, our data-augmentation strategy leads to simple, effective methods for posterior inference that: (1) circumvent the need for analytic approximations, numerical integration, or Metropolis--Hastings; and (2) outperform other known data-augmentation strategies, both in ease of use and in computational efficiency.  All methods, including an efficient sampler for the \Polya-Gamma distribution, are implemented in the R package \verb|BayesLogit|.

In the technical supplement appended to the end of the paper, we provide further details regarding the generation of \Polya-Gamma random variables; the empirical benchmarks reported in the main manuscript; and the extension of the basic data-augmentation framework to contingency tables and multinomial outcomes.

\end{abstract}

\begin{spacing}{1.1}

\section{Introduction}

Bayesian inference for the logistic regression model has long been recognized as a hard problem, due to the analytically inconvenient form of the model's likelihood function.  By comparison, Bayesian inference for the probit model is much easier, owing to the simple latent-variable method of \citet{albert:chib:1993} for posterior sampling.

In the two decades since the work of \citet{albert:chib:1993} on the probit model, there have been many attempts to apply the same missing-data strategy to the logit model \citep[e.g.][]{holmes:held:2006,fruhwirth-schnatter-fruhwirth-2010,gramacy:polson:2012}.  The results have been mixed.  Certainly many of these approaches have been used successfully in applied work.  Yet they all involve data-augmentation algorithms that are either approximate, or are significantly more complicated than the Albert/Chib method, as they involve multiple layers of latent variables.  Perhaps as a result, the Bayesian treatment of the logit model has not seen widespread adoption by non-statisticians in the way that, for example, the Bayesian probit model is used extensively in both political science and market research \citep[e.g.][]{rossi:allenby:mcculloch:2005,jackman:2009}.  The lack of a standard computational approach also makes it more difficult to use the logit link in the kind of complex hierarchical models that have become routine in Bayesian statistics.

In this paper, we present a new data-augmentation algorithm for Bayesian logistic regression.  Although our method involves a different missing-data mechanism from that of \citet{albert:chib:1993}, it is nonetheless a direct analogue of their construction, in that it is both exact and simple.  Moreover, because our method works for any binomial likelihood parametrized by log odds, it leads to an equally painless Bayesian treatment of the negative-binomial model for overdispersed count data.

This approach appeals to a new family of \Polya-Gamma distributions, described briefly here and constructed in detail in Section \ref{sec:polyagammadetails}.

\begin{mydefn}
A random variable $X$ has a \Polya-Gamma distribution with parameters $b > 0$ and $c \in \mathcal{R}$, denoted $X \sim \PG(b,c)$,  if
\begin{equation}
\label{eqn:PGdistribution1}
X \stackrel{D}{=} \frac{1}{2 \pi^2} \sum_{k=1}^{\infty} \frac{g_k}{(k-1/2)^2 + c^2/(4\pi^2)} \, ,
\end{equation}
where the $g_k \sim \Ga(b,1)$ are independent gamma random variables, and where $\stackrel{D}{=}$ indicates equality in distribution.
\end{mydefn}

Our main result (Theorem \ref{thm:PGmixture}, below) is that binomial likelihoods parametrized by log-odds can be represented as mixtures of Gaussians with respect to a \Polya-Gamma distribution.  The fundamental integral identity at the heart of our approach is that, for $b>0$,
\begin{equation}
 \label{eqn:logitlikelihood-general}
\frac{(e^{\psi})^a}{(1+e^{\psi})^b} = 2^{-b} e^{\kappa \psi} \int_0^{\infty} e^{-\omega \psi^2 / 2} \  p(\omega) \ d \omega \, ,
\end{equation}
where $\kappa = a-b/2$ and $\omega \sim \PG(b,0)$.  When $\psi = x^T \bbeta$ is a linear function of predictors, the integrand is the kernel of a Gaussian likelihood in $\bbeta$.  Moreover, as we will show below, the implied conditional distribution for $\omega$, given $\psi$, is also a \Polya-Gamma distribution. This suggests a simple strategy for Gibbs sampling across a wide class of binomial models: Gaussian draws for the main parameters, and \Polya-Gamma draws for a single layer of latent variables.

The success of this strategy depends upon the existence of a simple, effective way to simulate \Polya-Gamma random variables.  The sum-of-gammas representation in Formula (\ref{eqn:PGdistribution1}) initially seems daunting, and suggests only a na\"ive finite approximation.  But we describe a fast, exact \Polya-Gamma simulation method that avoids the difficulties that can result from truncating an infinite sum.  The method, which is implemented in the R package \verb|BayesLogit| \citep{bayeslogit:2013}, is an accept/reject sampler based on the alternating-series method of \citet{devroye:1986}.  For the basic $\mbox{PG}(1,c)$ case, the sampler is very efficient: it requires only exponential and inverse-Gaussian draws, and the probability of accepting a proposed draw is uniformly bounded below at 0.99919.  The method is also fully automatic, with no tuning needed to get optimal performance.  It is therefore sufficiently fast and reliable to be used as a black-box sampling routine in complex hierarchical models involving the logit link.

Many previous approaches have been proposed for estimating Bayesian logistic regression models.  This includes the Metropolis--Hastings method, along with many other latent-variable schemes that facilitate Gibbs sampling, all described below.  Thus a major aim of our paper is to demonstrate the efficiency of the \Polya-Gamma approach versus these alternatives across a wide range of circumstances.  We present evidence in support of two claims.
\begin{enumerate}
\item In simple logit models with abundant data and no hierarchical structure, the \Polya-Gamma method is a close second to the independence Metropolis-Hastings (MH) sampler, as long as the MH proposal distribution is chosen carefully.
\item In virtually all other cases, the \Polya-Gamma method is most efficient.
\end{enumerate}
The one exception we have encountered to the second claim is the case of a negative-binomial regression model with many counts per observation, and with no hierarchical structure in the prior.  Here, the effective sample size of the \Polya-Gamma method remains the best, but its effective sampling \textit{rate} suffers.  As we describe below, this happens because our present method for sampling $\PG(n,c)$ is to sum $n$ independent draws from $\PG(1,c)$; with large counts, this becomes a bottleneck.  In such cases, the method of \citet{fruhwirth-schnatter-etal-2009} provides a fast approximation, at the cost of introducing a more complex latent-variable structure.

This caveat notwithstanding, the \Polya-Gamma scheme offers real advantages, both in speed and simplicity, across a wide variety of structured Bayesian models for binary and count data.  In general, the more complex the model, and the more time that one must spend sampling its main parameters, the larger will be the efficiency advantage of the new method.  The difference is especially large for the Gaussian-process spatial models we consider below, which require expensive matrix operations.  We have also made progress in improving the speed of the \Polya-Gamma sampler for large shape parameters, beyond the method described in Section 4. These modifications lead to better performance in negative-binomial models with large counts.  They are detailed in \citet{windle:polson:scott:2013}, and have been incorporated into the latest version of our R package \citep{bayeslogit:2013}.

Furthermore, in a recent paper based on an early technical report of our method, \citet{choi:hobert:2013} have proven that the \Polya-Gamma Gibbs sampler for Bayesian logistic regression is uniformly ergodic.  This result has important practical consequences; most notably, it guarantees the existence of a central limit theorem for Monte Carlo averages of posterior draws.   We are aware of no similar result for any other MCMC-based approach to the Bayesian logit model.  Together with the numerical evidence we present here, this provides a strong reason to favor the routine use of the \Polya-Gamma method.

The paper proceeds as follows.  The \Polya-Gamma distribution is constructed in Section 2, and used to derive a data-augmentation scheme for binomial likelihoods in Section 3.  Section 4 describes a method for simulating from the \Polya-Gamma distribution, which we have implemented as a stand-alone sampler in the \verb|BayesLogit| R package.  Section 5 presents the results of an extensive benchmarking study comparing the efficiency of our method to other data-augmentation schemes.  Section \ref{sec:extensions} concludes with a discussion of some open issues related to our proposal.  Many further details of the sampling algorithm and our empirical study of its efficiency are deferred to a technical supplement.

\section{The \Polya-Gamma distribution}

\label{sec:polyagammadetails}

\subsection{The case $\PG(b,0)$}

The key step in our approach is the construction of the \Polya-Gamma
distribution.  We now describe this new family, deferring our method for simulating PG random variates to Section
\ref{sec:samplingPG}.

The \Polya-Gamma family of distributions, denoted $\PG(b,c)$, is a subset of the class
of infinite convolutions of gamma distributions.  We first focus on the $\PG(1,0)$ case, which is a carefully chosen element of the class of
infinite convolutions of exponentials, also know as \Polya{}
distributions \citep{bn:kent:sorensen:1982}.  The $\PG(1,0)$ distribution has Laplace transform
$\bbE\{\exp(- \omega t)\} = \cosh^{-1}(\sqrt{t/2})$.  Using this as a starting point, one may define the random variable $\omega \sim PG(b,0)$, $b > 0$, as the infinite
convolution of gamma distributions (hence the name \Polya-Gamma) that has
Laplace transform
\begin{equation}
\label{eqn:PGmgf1}
\bbE\{\exp(- \omega t)\} = \prod_{i=1}^t \Big(1 + \frac{t}{2 \pi^2
  (k-1/2)^2}\Big)^{-b} = \frac{1}{\cosh^b(\sqrt{t/2})} \, .
\end{equation}
The last equality is a consequence of the Weierstrass factorization theorem.  By inverting the Laplace transform, one finds that if $\omega \sim \PG(b,0)$, then it is equal in distribution to an infinite sum of gammas:
\begin{eqnarray*}
  \omega &\stackrel{D}{=}&  \frac{1}{2 \pi^2} \sum_{k=1}^{\infty} \frac{g_k }{(k-1/2)^2} \, ,
\end{eqnarray*}
where the $g_k \sim \mbox{Ga}(b,1)$ are mutually independent.

The $\PG(b,0)$ class of distributions is closely related to a subset of
distributions that are surveyed by
\citet{biane-etal-2001}.  This family of distributions, which we denote by $J^*(b)$, $b>0$, has close connections with the Jacobi Theta and Riemann Zeta functions, and with Brownian excursions.  Its Laplace transform is
\begin{equation}
\label{eqn:jacobi-def}
\bbE\{e^{-t \JJ(b)}\} = \cosh^{-b}(\sqrt{2t}) \, ,
\end{equation}
implying that $\PG(b,0) \stackrel{D}{=} \JJ(b) / 4$.

\subsection{The general $\PG(b,c)$ class}

The general $\PG(b,c)$ class arises through an exponential tilting of the
$\PG(b,0)$ density, much in the same way that a Gaussian likelihood combines with a Gamma prior for a precision.  Specifically, a $\PG(b,c)$ random variable has the probability density function
\begin{equation}
\label{eqn:polyagamma.generalpdf}
p(\omega \mid b, c) = \frac{\exp \left( -\frac{c^2}{2} \omega \right) p(\omega \mid b,0)}{ \E_{\omega} \left\{ \exp(-\frac{c^2}{2} \omega) \right\} } \, ,
\end{equation}
where $ p(\omega \mid b, 0)$ is the density of a $\PG(b,0)$ random variable.
The expectation in the denominator is taken with respect to the $\PG(b,0)$
distribution; it is thus $\cosh^{-b}(c/2)$ by (\ref{eqn:PGmgf1}), ensuring that
$p(\omega \mid b, c)$ is a valid density.

The Laplace transform of a $\PG(b,c)$ distribution may be calculated by
appealing to the Weierstrass factorization theorem again:
\begin{align}
\label{eqn:pg-mgf}
 \E_{\omega} \left\{ \exp\left( - \omega t \right) \right\} 
  &=  \frac{\cosh^{b} \left( \frac{c}{2} \right)}  {\cosh^{b} \left(
      \sqrt{\frac{c^2/2+t}{2}} \right) } \\
& = \prod_{k=1}^\infty \left( \frac{1+\frac{c^2/2}{2(k-1/2)^2\pi^2}}
  {1+\frac{c^2/2+t}{2(k-1/2)^2\pi^2}} \right)^{b} \notag \\
 & = \prod_{k=1}^\infty (1+ d_k^{-1} t )^{-b} \; , \quad  {\rm where} \;  d_k =
 2\left(k-\frac{1}{2}\right)^2\pi^2 + c^2/2 \; . \notag
\end{align}

Each term in the product is recognizable as the Laplace transform of a gamma
distribution.  We can therefore write a $\PG(b,c)$ as an infinite convolution of
gamma distributions,
\begin{align*}
\omega \stackrel{D}{=} \sum_{k=1}^\infty \frac{\Ga(b,1)}{d_k} = 
\frac{1}{2 \pi^2} \sum_{k=1}^\infty \frac{\Ga(b,1)}{(k-\frac{1}{2})^2  + c^2/(4 \pi^2) } \, ,
\end{align*}
which is the form given in Definition 1.

\subsection{Further properties}

The density of a \Polya-Gamma random variable can be expressed as an
alternating-sign sum of inverse-Gaussian densities.  This fact plays a crucial role in our method for simulating \Polya-Gamma draws.  From the characterization
of $\JJ(b)$ density given by \citet{biane-etal-2001}, we know that the $
PG(b,0)$ distribution has density
$$
f( x \mid b, 0) = \frac{2^{b-1}}{\Gamma(b)} \sum_{n=0}^\infty (-1)^n \frac{\Gamma (n+b)}{\Gamma(n+1)} 
 \frac{(2n+b)}{\sqrt{2\pi x^3}} e^{ - \frac{(2n+b)^2}{8 x} } \, .
$$
The density of $ PG(b,z)$ distribution is then computed by an exponential tilt and a renormalization:
$$
f( x \mid b, c ) =  \{\cosh^{b}(c/2)\} \frac{2^{b-1}}{\Gamma(b)} \sum_{n=0}^\infty (-1)^n \frac{\Gamma (n+b)}{\Gamma(n+1)} 
 \frac{(2n+b)}{\sqrt{2\pi x^3}} e^{ - \frac{(2n+b)^2}{8 x} - \frac{c^2}{2} x } \, .
$$
Notice that the normalizing constant is known directly from the Laplace
transform of a $\PG(b,0)$ random variable.

A further useful fact is that all finite moments of a \Polya-Gamma random
variable are available in closed form.  In particular, the expectation may be
calculated directly.  This allows the \Polya-Gamma scheme to be used in EM
algorithms, where the latent $\omega$'s will form a set of complete-data
sufficient statistics for the main parameter.  We arrive at this result by appealing to the Laplace transform of $\omega \sim \PG(b,c)$.
Differentiating (\ref{eqn:pg-mgf}) with respect to $t$, negating, and evaluating
at zero yields
\[
\bbE(\omega) = \frac{b}{2c} \tanh(c/2) = \frac{b}{2c} \left( \frac{e^c - 1}{1+e^c} 
 \right) \, .
\]

Lastly, the \Polya-Gamma class is closed under convolution for random variates
with the same scale (tilting) parameter.  If $\omega_1 \sim \PG(b_1, z)$ and $\omega_2
\sim \PG(b_2, z)$ are independent, then $\omega_1 + \omega_2 \sim PG(b_1 + b_2,
z)$.  This follows from the Laplace transform.  We will employ this property
later when constructing a \Polya-Gamma sampler.

\section{The data-augmentation strategy}

\subsection{Main result}

\label{subsec:proofpgmixture}

The \Polya-Gamma family has been carefully constructed to yield a simple Gibbs sampler for the Bayesian logistic-regression model. The two differences from the \citet{albert:chib:1993} method for probit regression are that the posterior distribution is a scale mixture, rather than location mixture, of Gaussians; and that Albert and Chib's truncated normals are replaced by \Polya-Gamma latent variables.

To fix notation: let $y_i$ be the number of successes, $n_i$ the number of trials, and $x_i = (x_{i1}, \ldots, x_{ip})$ the vector of regressors for observation $i \in \{1, \ldots, N\}$.  Let $y_i \sim \mbox{Binom}(n_i, 1/\{1+e^{-\psi_i}\})$, where $\psi_i = x_i^T \bbeta$ are the log odds of success.  Finally, let $\bbeta$ have a Gaussian prior, $\bbeta \sim \N(b, B)$.  To sample from the posterior distribution using the \Polya-Gamma method, simply iterate two steps:
\begin{eqnarray*}
(\omega_i \mid \bbeta) &\sim& \PG(n_i, x_i^T \bbeta) \\
(\bbeta \mid y, \omega) &\sim& \N(m_{\omega}, V_{\omega}) \, ,
\end{eqnarray*}
where
\begin{eqnarray*}
V_{\omega} &=& (X^T \Omega X + B^{-1})^{-1} \\
m_{\omega} &=& V_{\omega} (X^T \kappa  + B^{-1} b) \, ,
\end{eqnarray*}
where $\kappa = (y_1 - n_1/2, \ldots, y_N - n_N/2)$, and $\Omega$ is the diagonal matrix of $\omega_i$'s.

We now derive this sampler, beginning with a careful statement and proof of the integral identity mentioned in the introduction.

\begin{theorem}
\label{thm:PGmixture}
Let $p(\omega)$ denote the density of the random variable $\omega \sim \PG(b,0)$, $b>0$.  Then the following integral identity holds for all $a \in \mathbb{R}$:
\begin{equation}
\frac{(e^{\psi})^a}{(1+e^{\psi})^b} = 2^{-b} e^{\kappa \psi} \int_0^{\infty} e^{-\omega \psi^2 / 2} \  p(\omega) \ d \omega \, , \label{eqn:pgprior}
\end{equation}
where $\kappa = a - b/2$.

Moreover, the conditional distribution
$$
p(\omega \mid \psi) = \frac{e^{-\omega \psi^2 / 2} \  p(\omega) } {\int_0^{\infty} e^{-\omega \psi^2 / 2} \  p(\omega) \ d \omega} \, ,
$$
which arises in treating the integrand in (\ref{eqn:pgprior}) as an unnormalized joint density in $(\psi, \omega)$, is also in the \Polya-Gamma class: $(\omega \mid \psi) \sim \PG (b, \psi)$.
\end{theorem}

\begin{proof}
  Appealing to (\ref{eqn:PGmgf1}), we may
  write the lefthand side of (\ref{eqn:pgprior}) as
\begin{eqnarray*}
\frac{( e^{\psi})^{a} } {( 1 + e^{\psi})^{b} }  
&=&  \frac{  2^{-b} \exp \{ \kappa \psi \} } { \cosh^{b} (\psi/2)}  \\ \\
&=& 2^{-b} e^{\kappa \psi} \ \E_{\omega} \{ \exp(-\omega \psi^2/2 \} \, ,
\end{eqnarray*}
where the expectation is taken with respect to $\omega \sim \PG(b, 0)$, and
where $\kappa = a- b/2$.

Turn now to the conditional distribution
$$
p(\omega \mid \psi) = \frac{e^{-\omega \psi^2 / 2} \  p(\omega) } {\int_0^{\infty} e^{-\omega \psi^2 / 2} \  p(\omega) \ d \omega} \, ,
$$
where $p(\omega)$ is the density of the prior, $\PG(b,0)$.  This is of the same
form as (\ref{eqn:polyagamma.generalpdf}), with $\psi = c$.  Therefore $(\omega
\mid \psi) \sim \PG(b,\psi)$.

\end{proof}

To derive our Gibbs sampler, we appeal to Theorem \ref{thm:PGmixture} and write the likelihood contribution of observation $i$ as
\begin{eqnarray*}
L_i(\bbeta) &=& \frac{\{\exp(x_i^T \bbeta)\}^{y_i}}{1+\exp(x_i^T \bbeta)} \\
&\propto& \exp(\kappa_i x_i^T \bbeta) \ \int_0^{\infty} \exp\{-\omega_i (x_i^T \bbeta)^2 /2 \} \ p(\omega_i \mid n_i, 0) \, ,
\end{eqnarray*}
where $\kappa_i = y_i - n_i/2$, and where $p(\omega_i \mid n_i, 0)$ is the density of a \Polya-Gamma random variable with parameters $(n_i,0)$.

Combining the terms from all $n$ data points gives the following expression for the conditional posterior of $\bbeta$, given $\omega = (\omega_1, \ldots, \omega_N)$:
\begin{eqnarray*}
p(\bbeta \mid \omega, y) \propto p(\bbeta) \prod_{i=1}^N L_i(\bbeta \mid \omega_i) &=&  p(\bbeta) \prod_{i=1}^N  \exp\left\{\kappa_i x_i^T \bbeta -\omega_i (x_i^T \bbeta)^2 /2 \right\}  \\
&\propto& p(\bbeta)  \prod_{i=1}^N  \exp\left\{ \frac{\omega_i}{2} (x_i^T \bbeta - \kappa_i/\omega_i)^2 \right\} \\
&\propto& p(\bbeta)  \exp \left\{ -\frac{1}{2} (z - X \bbeta)^T \Omega (z - X \bbeta) \right\} \, ,
\end{eqnarray*}
where $z = (\kappa_1/\omega_1, \ldots, \kappa_n/\omega_N)$, and where $\Omega = \mbox{diag}(\omega_1, \ldots, \omega_N)$.  This is a conditionally Gaussian likelihood in $\bbeta$, with working responses $z$, design matrix $X$, and diagonal covariance matrix $\Omega^{-1}$.  Since the prior $p(\bbeta)$ is Gaussian, a simple linear-model calculation leads to the Gibbs sampler defined above.

\subsection{Existing data-augmentation schemes}

A comparison with the methods of \citet{holmes:held:2006} and
\citet{fruhwirth-schnatter-fruhwirth-2010} clarifies how the \Polya-Gamma method
differs from previous attempts at data augmentation. Both of these methods attempt to replicate the missing-data mechanism of \citet{albert:chib:1993}, where the outcomes $y_i$ are assumed to
be thresholded versions of an underlying continuous quantity $z_i$.  For simplicity, we assume that $n_i=1$ for all observations, and that $y_i$ is either 0 or 1.  Let
\begin{eqnarray}
y_i &=& 
\left\{
\begin{array}{r r}
1 \, , \quad z_i \geq 0 \nonumber \\
0 \, , \quad z_i < 0
\end{array}
\right. \\
z_i &=& x_i^T \bbeta + \epsilon_i \; , \quad \epsilon_i  \sim \Lo(1) \, ,\label{eqn:randutillogit}
\end{eqnarray}
where $\epsilon_i \sim \Lo(1)$ has a standard logistic distribution. Upon
marginalizing over the $z_i$, often called the latent utilities, the original binomial likelihood is recovered.

Although (\ref{eqn:randutillogit}) would initially seem to be a direct parallel with \citet{albert:chib:1993}, it does
not lead to an easy method for sampling from the posterior distribution of
$\bbeta$. This creates additional complications compared to the probit case.  The standard approach has been to add another layer of auxiliary variables to handle the logistic error model on the latent-utility scale.  One strategy is to represent the logistic distribution as a normal-scale mixture \citep{holmes:held:2006}:
\begin{eqnarray*}
(\epsilon_i \mid \phi_i)  &\sim& \N(0, \phi_i) \\
\phi_i &=& (2 \lambda_i^2) \; , \quad \lambda_i \sim \KS(1) \, ,
\end{eqnarray*}
where $\lambda_i$ has a Kolmogorov--Smirnov distribution
\citep{andrews:mallows:1974}.  Alternatively, one may approximate the logistic error term as
a discrete mixture of normals \citep{fruhwirth-schnatter-fruhwirth-2010}:
\begin{eqnarray*}
(\epsilon_i \mid \phi_i)  &\sim& \N(0, \phi_i) \\
\phi_i &\sim& \sum_{k=1}^K w_k \delta_{\phi^{(k)}} \, ,
\end{eqnarray*}
where $\delta_\phi$ indicates a Dirac measure at $\phi$. The weights $w_k$ and
the points $\phi^{(k)}$ in the discrete mixture are fixed for a given choice of
$K$ so that the Kullback--Leibler divergence from the true distribution of the
random utilities is minimized. \citet{fruhwirth-schnatter-fruhwirth-2010} find
that the choice of $K=10$ leads to a good approximation, and list the optimal
weights and variances for this choice.

In both cases, posterior sampling can be done in two blocks, sampling the complete conditional of $\bbeta$ in one block
and sampling the joint complete conditional of both layers of auxiliary
variables in the second block. The discrete mixture of normals is an approximation, but it outperforms the
scale mixture of normals in terms of effective sampling rate, as it is much faster.

One may also arrive at the hierarchy above by manipulating the random utility-derivation of \citet{mcfadden:1974}; this involves the difference of random
utilities, or ``dRUM,'' using the term of \cite{fruhwirth-schnatter-fruhwirth-2010}.
The dRUM representation is superior to the random utility approach explored in
\cite{fruhwirth-schnatter-fruhwirth-2007}. Further work by
\cite{fussl-etal-2011-slides} improves the approach for binomial logistic models. In this extension, one must use a table
of different weights and variances representing different normal mixtures, to
approximate a finite collection of type-III logistic distributions, and interpolate within this table to approximate the entire family.

Both \cite{albert:chib:1993} and \citet{obrien-dunson-2004} suggest another approximation: namely, the use of a Student-$t$ link function as a close substitute for the logistic link. But this also introduces a second layer of latent  variables, in that the Student-$t$ error model for $z_i$ is represented as a scale mixture of normals.

\begin{figure}[t]
\begin{center}
\includegraphics[width=4.5in]{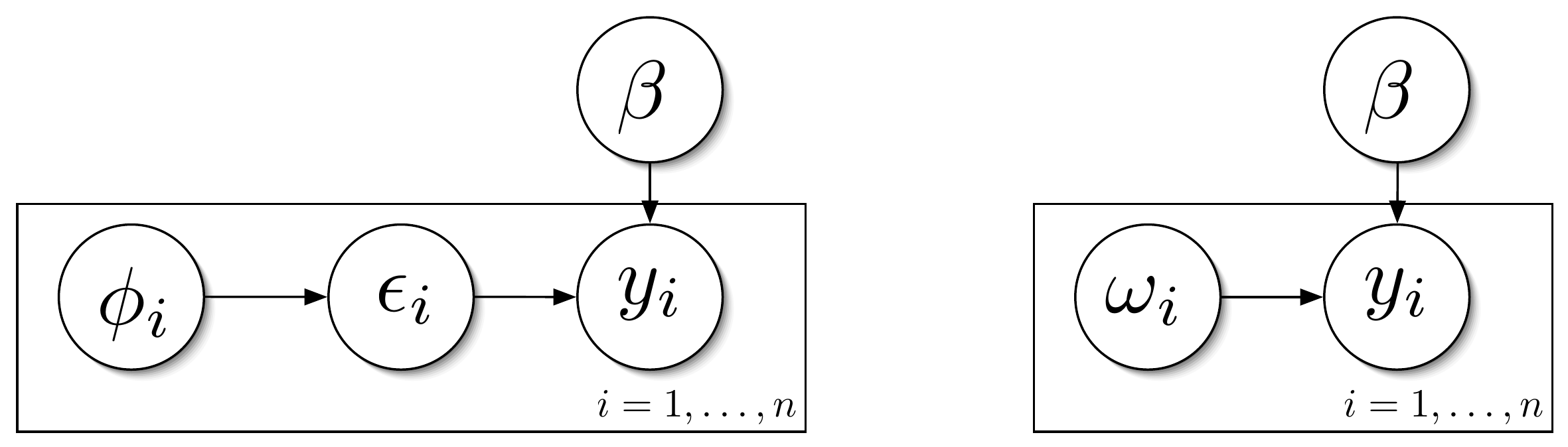}
\caption{\label{fig:logitDAG} Directed acyclic graphs depicting two
  latent-variable constructions for the logistic-regression model: the
  difference of random-utility model of \citet{holmes:held:2006} and
  \citet{fruhwirth-schnatter-fruhwirth-2010}, on the left; versus our direct
  data-augmentation scheme, on the right.}
\end{center}
\end{figure}

Our data-augmentation scheme differs from each of these approaches in several ways. First, it does not appeal directly to the random-utility interpretation of the logit model.  Instead, it represents the logistic CDF as a mixture with respect to an infinite convolution of gammas.  Second, the method is
exact, in the sense of making draws from the correct joint posterior
distribution, rather than an approximation to the posterior that arises out of an approximation to the link function. Third, like the \citet{albert:chib:1993} method, it requires
only a single layer of latent variables.

A similar approach to ours is that of \citet{gramacy:polson:2012}, who
propose a latent-variable representation of a powered-up version of the
logit likelihood \citep[c.f.][]{Polson:Scott:2011a}. This representation is useful for obtaining classical
penalized-likelihood estimates via simulation, but for the ordinary logit model it leads
to an improper mixing distribution for the latent variable. This requires
modifications of the basic approach that make simulation difficult in the
general logit case. As our experiments show, the method does not seem to be
competitive on speed grounds with the \Polya-Gamma representation, which results in a proper
mixing distribution for all common choices of $a_i,b_i$ in
(\ref{eqn:logitlikelihood-general}).

For negative-binomial regression, \citet{fruhwirth-schnatter-etal-2009} employ
the discrete-mixture/table-interpolation approach, like that used by
\citet{fussl-etal-2011-slides}, to produce a tractable data augmentation scheme.
In some instances, the \Polya-Gamma approach outperforms this method; in
others, it does not. The reasons for this discrepancy can be explained by
examining the inner workings of our \Polya-Gamma sampler, discussed in Section 4.

\subsection{Mixed model example}

We have introduced the \Polya-Gamma method in the context of a binary logit model.  We do this with the understanding that, when data are abundant, the Metropolis--Hastings algorithm with independent proposals will be efficient, as asymptotic theory suggests that a normal approximation to the posterior distribution will become very accurate as data accumulate.  This is well understood among Bayesian practitioners \citep[e.g.][]{carlin:1992,gelman:carlin:stern:rubin:2004}.

But the real advantage of data augmentation, and the \Polya-Gamma technique in particular, is that it becomes easy to construct and fit more complicated models.  For instance, the \Polya-Gamma method trivially accommodates mixed models, factor models, and models with a spatial or dynamic structure.  For most problems in this class, good Metropolis--Hastings samplers are difficult to design, and at the very least will require ad-hoc tuning to yield good performance.

\begin{figure}
\begin{center}
\includegraphics[width=6.0in]{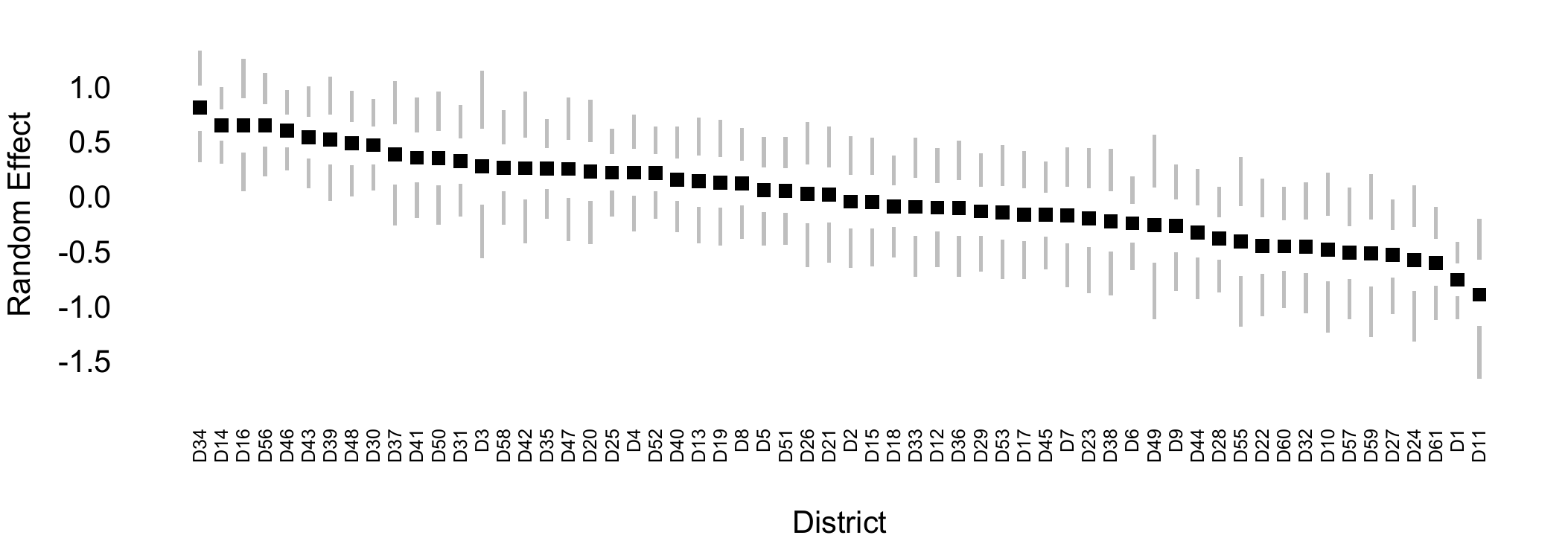}
\caption{\label{fig:mm-ranef} Marginal posterior distribution of random
  intercepts for each district found in a Bangladeshi contraception survey.  For 10,000 samples after 2,000
  burn-in, median ESS=8168 and median ESR=59.88 for the PG method.  Grey/white bars: 90\%/50\% posterior credible intervals.  Black dots: posterior means.}
\end{center}
\end{figure}

Several relevant examples are considered in Section 5.  But as an initial
illustration of the point, we fit a binomial logistic mixed model using the data
on contraceptive use among Bangladeshi women provided by the R package
\texttt{mlmRev} \citep{mlmrev-2011}.  The data comes from a Bangladeshi survey
whose predictors include a woman's age, the number of children at the time of
the survey, whether the woman lives in an urban or rural area, and a more
specific geographic identifier based upon the district in which the woman
resides.  Some districts have few observations and district 54 has no
observations; thus, a mixed model is necessary if one wants to include this
effect.  The response identifies contraceptive use.  We fit the mixed model
\begin{eqnarray*}
y_{ij} & \sim & \mbox{Binom}(1,p_{ij}) \; , \quad p_{ij}= \frac{e^{\psi_{ij}}}{1+e^{\psi_{ij}}},
\\
\psi_{ij} & = & m + \delta_j + x_{ij}' \beta, \\
\delta_j & \sim & N(0, 1/\phi), \\
m & \sim & N(0, \kappa^2 / \phi),
\end{eqnarray*}
where $i$ and $j$ correspond to the $i$th observation from the $j$th district.
The fixed effect $\beta$ is given a $N(0, 100 I)$ prior while the precision
parameter $\phi$ is given $\Ga(1,1)$ prior.  We take $\kappa \ra \infty$ to
recover an improper prior for the global intercept $m$.  Figure
\ref{fig:mm-ranef} shows the box plots of the posterior draws of the random
intercepts $m + \delta_j$.  If one does not shrink these random intercepts to a
global mean using a mixed model, then several take on unrealistic values due to
the unbalanced design.

We emphasize that there are many ways to model this data, and that we do not intend our analysis to be taken as definitive.  It is merely a proof of concept, showing how various aspects of Bayesian hierarchical modeling---in this case, models with both fixed and random effects---can be combined routinely with binomial likelihoods using the \Polya-Gamma scheme.  Together these changes require just a few lines of code and a few extra seconds of runtime compared to the non-hierarchical logit model.  A posterior draw of 2,000 samples for this data set takes 26.1 seconds for a binomial logistic
regression, versus 27.3 seconds for a binomial logistic mixed model.  As seen
in the negative binomial examples below, one may also painlessly incorporate a
more complex prior structure using the \Polya-Gamma technique.  For instance, if
given information about the geographic location of each district, one could place
spatial process prior upon the random offsets $\{\delta_j\}$.

\section{Simulating \Polya-Gamma random variables}
\label{sec:samplingPG}

\subsection{The PG(1,z) sampler}

All our developments thus far require an efficient method for sampling
\Polya-Gamma random variates.  In this section, we derive such a method, which
is implemented in the R package \texttt{BayesLogit}.  We focus chiefly on simulating PG(1,z) efficiently, as this is most relevant to the binary logit model.

First, observe that one may sample \Polya-Gamma random variables na\"ively (and approximately) using
the sum-of-gammas representation in Equation (\ref{eqn:PGdistribution1}).  But
this is slow, and involves the potentially dangerous step of truncating an
infinite sum.

We therefore construct an alternate, exact method by extending
the approach of \citet{devroye:2009} for simulating $\JJ(1)$ from
(\ref{eqn:jacobi-def}).  The distribution $\JJ(1)$ is related to the Jacobi
theta function, so we call $\JJ(1)$ the Jacobi distribution.  One may define an
exponentially tilted Jacobi distribution $\JJ(1,z)$ via the density
\begin{equation}
\label{eqn:jacobitilt} f(x \mid z) = \cosh(z) \ e^{-x z^2 /2} \ f(x) \, ,
\end{equation} 
where $f(x)$ is the density of $\JJ$(1).  The $PG(1,z)$ distribution is related
to $\JJ(1, z)$ through the rescaling
\begin{equation}
\label{eqn:pg-jacobi}
PG(1,z) = \frac{1}{4} \JJ(1, z/2).
\end{equation}

\citet{devroye:2009} develops an efficient $\JJ(1,0)$ sampler.  Following this
work, we develop an efficient sampler for an exponentially tilted $\JJ$ random variate,
$\JJ(1,z)$.  In both cases, the density of interest can be written as an
infinite, alternating sum that is amenable to the series
method described in Chapter IV.5 of \citet{devroye:1986}.  Recall that a random variable with density $f$ may be sampled by the accept/reject
algorithm by: (1) proposing $X$ from a density $g$; (2) drawing $U \sim \mcU(0,
cg(X))$ where $\|f/g\|_{\infty} \leq c$; and (3) accepting $X$ if $U \leq f(X)$
and rejecting $X$ otherwise.  When $f(x) = \sum_{n=0}^{\infty} (-1)^n a_n(x)$
and the coefficients $a_n(x)$ are decreasing for all $n \in \mathbb{N}_0$, for
fixed $x$ in the support of $f$, then the partial sums, $S_n(x) = \sum_{i=0}^n
(-1)^i a_i(x)$, satisfy
\begin{equation}
\label{eqn:partial-sum}
S_0(x) > S_2(x) > \cdots > f(x) > \cdots > S_3(x) > S_1(x).
\end{equation}
In that case, step (3) above is equivalent to accepting $X$ if $U \leq S_i(X)$ for some odd $i$, and rejecting $X$ if $U > S_i(X)$ for some even $i$.  Moreover, the partial sums $S_i(X)$ can be calculated iteratively.  Below we show that for
the $J^*(1,z)$ distribution the algorithm will accept with high probability upon
checking $U \leq S_1(X)$.

The Jacobi density has two alternating-sum representations,
$\sum_{n=0}^\infty (-1)^n a_n^{L}(x)$ and $\sum_{n=0}^\infty (-1)^n a_i^{R}(x)$,
neither of which satisfy (\ref{eqn:partial-sum}) for all $x$ in the support of
$f$.  However, each satisfies (\ref{eqn:partial-sum}) on an interval.  These two intervals, respectively denoted $I_L$ and $I_R$, satisfy $I_L \cup I_R = (0,\infty)$ and $I_L \cap I_R \neq \emptyset$.  Thus,
one may pick $t > 0$ and define the piecewise coefficients
\begin{numcases}{a_n(x) =}
\label{jacobi:dL}
\pi (n+1/2) \; \left(\frac{2}{\pi x}\right)^{3/2}
  \exp \left\{ - \frac{2(n+1/2)^2}{x} \right\} \, , & $0 < x \leq t$, \\
\label{jacobi:dR}
\pi (n+1/2) \;
  \exp \left\{ -\frac{(n+1/2)^2 \pi^2}{2} x \right\} \, , & $x > t$,
\end{numcases}
so that $f(x) = \sum_{n=0}^\infty (-1)^n a_n(x)$ satisfies the partial sum
criterion (\ref{eqn:partial-sum}) for $x > 0$.  Devroye shows that the best
choice of $t$ is near $0.64$.

Employing (\ref{eqn:jacobitilt}), we now see that the $\JJ(1,z)$ density can be
written as an infinite, alternating sum $f(x|z) = \sum_{n=0}^\infty (-1)^n
a_n(x|z)$, where
\[
a_n(x|z) = \cosh(z) \exp \left\{ - \frac{z^2 x}{2} \right\} a_n(x) \, .
\]
This satisfies (\ref{eqn:partial-sum}), as $a_{n+1}(x|z) / a_{n}(x|z) =
a_{n+1}(x) / a_n(x)$.  Since $a_0(x|z) \geq f(x|z)$, the first term of the series provides a natural
proposal:
\begin{equation}
\label{eqn:prop-kernel}
c(z) \, g(x|z) = \frac{\pi}{2} \cosh(z)
\begin{dcases}
\left( \frac{2}{\pi x} \right)^{3/2} 
\exp \left\{ - \frac{z^2 x}{2} - \frac{1}{2x} \right\} \, ,&  0 < x \leq t, \\
\exp \left\{ - \left(\frac{z^2}{2} + \frac{\pi^2}{8} \right) x \right\} \, , & x > t.
\end{dcases}
\end{equation}
Examining these two kernels, one finds that $X \sim g(x|z)$ may be sampled
from a mixture of an inverse-Gaussian and an exponential:
\[
X \sim
\begin{cases}
IG(|z|^{-1}, 1) \bbI_{(0,t]}& \textmd{ with prob. } p / (p+q) \\
Ex(-z^2/2+\pi^2/8) \bbI_{(t,\infty)} & \textmd{ with prob. } q / (p+q)
\end{cases}
\]
where $p(z) = \int_0^t c(z) \, g(x|z) dx$ and $q(z) = \int_t^\infty c(z) \,
g(x|z) dx$.  Note that we are implicitly suppressing the dependence of $p,q,c$,
and $g$ upon $t$.

With this proposal in hand, sampling $\JJ(1,z)$ proceeds as
follows:
\begin{enumerate}
\item Generate a proposal $X \sim g(x|z)$.
\item Generate $U \sim \mcU(0, c(z)g(X|z))$.
\item Iteratively calculate $S_n(X|z)$, starting at $S_1(X|z)$, until $U \leq
  S_n(X|z)$ for an odd $n$ or until $U > S_n(X|z)$ for an even $n$.
\item Accept $X$ if $n$ is odd; return to step 1 if $n$ is even.
\end{enumerate}
To sample $Y \sim PG(1,z)$, draw $X \sim J^*(1,z/2)$ and then let $Y = X / 4$.
The details of the implementation, along with pseudocode, can be found in the
technical supplement.

\subsection{Analysis of acceptance rate}

This $\JJ(1,z)$ sampler is very efficient. The parameter $c=c(z,t)$ found in
(\ref{eqn:prop-kernel}) characterizes the average number of proposals we expect
to make before accepting. Devroye shows that in the case of $z=0$, one can pick
$t$ so that $c(0,t)$ is near unity. We extend this result to non-zero tilting
parameters and calculate that, on average, the $\JJ(1,z)$ sampler rejects no
more than 9 out of every 10,000 draws, regardless of $z$.

\begin{proposition}
\label{prop:accept-rate}
Define
\[
p(z,t) = \int_0^t \frac{\pi}{2} \cosh(z) \exp \left\{ -\frac{z^2 x}{2} \right\} a_0^L(x) dx,
\]  
\[
q(z,t) = \int_t^\infty \frac{\pi}{2} \cosh(z) \exp \left\{ -\frac{z^2 x}{2} \right\} a_0^R(x) dx \, .
\]
The following facts about the \Polya-Gamma rejection sampler hold.
\begin{enumerate}
\item The best truncation point $t^*$ is independent of $z \geq 0$.  
\item For a fixed truncation point $t$, $p(z,t)$ and $q(z,t)$ are continuous,
  $p(z,t)$ decreases to zero as $z$ diverges, and $q(z,t)$ converges to $1$ as
  $z$ diverges.  Thus $c(z,t) = p(z,t) + q(z,t)$ is continuous and converges to $1$ as $z$
  diverges.
\item For fixed $t$, the average probability of accepting a draw, $1/c(z,t)$, is
  bounded below for all $z$.  For $t^*$, this bound to five digits is 0.99919, which is attained at $z \simeq 1.378$.
\end{enumerate}
\end{proposition}

\begin{proof}
  We consider each point in turn.  Throughout, $t$ is assumed to be in the
  interval of valid truncation points, $I_L \cap I_R$.

\begin{enumerate}

\item We need to show that for fixed $z$, $c(z,t) = p(z,t) + q(z,t)$ has a
  maximum in $t$ that is independent of $z$.  For fixed $z \geq 0$, $p(z,t)$ and
  $q(z,t)$ are both differentiable in $t$.  Thus any extrema of $c$ will occur
  on the boundary of the interval $I_L \cap I_R$, or at the critical points for
  which $\frac{\del c}{\del t} = 0$; that is, $t \in I_L \cap I_R$, for which
  \[
  \cosh(z) \exp \Big\{ - \frac{z^2}{2} t \Big\} [a_0^L(t) - a_0^R(t)] = 0.
  \]
  The exponential term is never zero, so an interior critical point must satisfy
  $a_0^L(t) - a_0^R(t) = 0$, which is independent of $z$.  Devroye shows there
  is one such critical point, $t^* \simeq 0.64$, and that it corresponds to a
  maximum.

\item Both $p$ and $q$ are integrals of recognizable kernels.  Rewriting the
  expressions in terms of the corresponding densities and integrating yields
  \[
  p(z,t) = \cosh(z) \frac{\pi}{2} \frac{1}{y(z)} \exp \Big\{ - y(z) t \Big\},
  \;\; \;
  y(z) = \frac{z^2}{2} + \frac{\pi^2}{8},
  \]
  and
  \[
  q(z,t) = (1+e^{-2z}) \Phi_{IG}(t | 1/z, 1) \, ,
  \]
  where $\Phi_{IG}$ is the cumulative distribution function of an $IG(1/z,1)$
  distribution.

  One can see that $p(z,t)$ is eventually decreasing in $z$ for fixed $t$ by
  noting that the sign of $\frac{\del p}{\del z}$ is determined by
  \[
  \tanh(z) - \frac{z}{\frac{z^2}{2} + \frac{\pi^2}{8}} - z t\, ,
  \]
  which is eventually negative.  (In fact, for the $t^*$ calculated above it
  appears to be negative for all $z \geq 0$, which we do not prove here.)
  Further, $p(z,t)$ is continuous in $z$ and converges to $0$ as $z$ diverges.
  
  To see that $q(z,t)$ converges to $1$, consider a Brownian motion $(W_s)$
  defined on the probability space $(\Omega, \mcF, \bbP)$ and the subsequent
  Brownian motion with drift $X_s^{z} = z s + W_s$.  The stopping time
  \(
  T^{z} = \inf \{ s > 0 | X_s^{z} \geq 1 \}
  \)
  is distributed as $IG(1/z, 1)$ and 
  \(
  \bbP( T^{z} < t ) = \bbP( \max_{s \in [0,t]} X_s^z \geq 1 ) \, .
  \)

  Hence $\bbP(T^z < t)$ is increasing and $\lim_{z \ra \infty} \bbP(T^z < t) =
  1$, ensuring that $q(z,t) \propto (1+e^{-2z}) \bbP(T^z < t)$ converges to 1 as
  $z \ra \infty$ as well.  Continuity follows by considering the cumulative
  distribution $\bbP(T^{z} < t) = \Phi\{(zt-1)/\sqrt{t}\} - \exp(2 z t)
  \Phi\{(-1-zt)/\sqrt{t}\}$, which is a composition of continuous functions in
  $z$.

  By the continuity and tail behavior of $p$ and $q$, it follows that $c(z,t) =
  p(z,t) + q(z,t)$, for fixed $t$, is continuous for all $z$ and converges to
  $1$ as $z$ diverges.  Further $c(z,t) \geq 1$ since the target density and
  proposal density satisfy $f(x|z) \leq c(z,t) g(x|z)$ for all $x \geq 0$.
  Thus, $c$ takes on its maximum over $z$.

\item Since, for each $t$, $c(z,t)$ is bounded above in $z$, we know that
  $1/c(z,t)$ is bounded below above zero.  For $t^*$, we numerically calculate
  that $1/c(z,t^*)$ attains its minimum $0.9991977$ at $z \simeq 1.378$; thus,
  $1 / c(z,t^*) > 0.99919$ suggesting that no more than 9 of every 10,000 draws
  are rejected on average.
  
\end{enumerate}

\end{proof}

Since $t^*$ is the best truncation point regardless of $z$, we will assume that
the truncation point has been fixed at $t^*$ and suppress it from the notation.

\subsection{Analysis of tail probabilities}

Proposition \ref{prop:accept-rate} tells us that the sampler rarely rejects a
proposal.  One possible worry, however, is that the algorithm might calculate many terms in the sum before deciding to accept or reject, and that the sampler would be slow despite rarely rejecting.

Happily, this is not the case, as we now prove.  Suppose one samples $X \sim J^*(1,z)$.  Let $N$ denote the total number of
proposals made before accepting, and let $L_n$ be the number of partial sums
$S_i$ ($i=1, \ldots, L_n$) that are calculated before deciding to accept or reject
proposal $n \leq N$.  A variant of theorem 5.1 from \cite{devroye:1986} employs
Wald's equation to show that that
\(
\bbE[\sum_{n=1}^N L_n] = \sum_{i=0}^\infty \int_{0}^\infty a_i(x|z) dx.
\)
For the worst enclosing envelope, $z \simeq 1.378$, $\bbE[N] = 1.0016$; that is,
on average, one rarely calculates anything beyond $S_1$ of the first proposal.
A slight alteration of this theorem gives a more precise sense of how many terms
in the partial sum must be calculated.

\begin{proposition}
  When sampling $X \sim J^*(1,z)$, the probability of deciding to accept or
  reject upon checking the $n$th partial sum $S_n$, $n \geq 1$, is
\[
\frac{1}{c(z)} \int_0^\infty \left\{ a_{n-1}(x|z) - a_n(x|z) \right\} \ dx.
\]
\end{proposition}

\begin{proof}
  Let $L$ denote the number of partials sums that are calculated before
  accepting or rejecting the proposal.  That is, a proposal, $X$, is generated;
  $U$ is drawn from $\mcU(0, a_0(X|z))$; and $L$ is the smallest natural number
  $n \in \bbN$ for which $U \leq S_n$ if $n$ is odd or $U > S_n$ if $n$ is even,
  where $S_n$ denotes $S_n(X|z)$.  But since $L$ is the smallest $n$ for which
  this holds, $S_{L-2} < U \leq S_L$ when $L$ is odd and $S_{L} < U \leq
  S_{L-2}$ when $L$ is even.  Thus, the algorithm accepts or rejects if and only
  if $U \in K_L(X|z)$ where
  \[
  K_n(x|z) = 
  \begin{cases}
    (S_{n-2}(x|z), S_n(x|z)], & \text{odd } n \\
    (S_n(x|z), S_{n-2}(x|z)], & \text{even } n.
  \end{cases}
  \]
  In either case, $|K_n(x|z)| = a_{n-1}(x|z) - a_n(x|z)$.  Thus
  \[
  \bbP(L = n | X=x) = \frac{a_{n-1}(x|z) - a_n(x|z)}{a_0(x|z)}.
  \]
  Marginalizing over $x$ yields
  \[
  \bbP(L=n) = \frac{1}{c(z)} \int_{0}^\infty \left\{ a_{n-1}(x|z) - a_n(x|z) \right\} dx.
  \]
\end{proof}

Since each coefficient $a_n$ is the piecewise composition of an inverse Gaussian
kernel and an exponential kernel, these integrals may be evaluated.  In
particular,
\[
a_n(x|z)
=
\cosh(z) 
\begin{cases}
  2 e^{-(2n+1) z} \; p_{IG}(x | \mu_n(z), \lambda_n), & x < t \\
  \pi \Big(n + \frac{1}{2} \Big) \frac{1}{y_n(z)} \; p_{\mcE}(x | y_n(z)), & x
  \geq t \, ,
\end{cases}
\] 
where $\mu_n(z) = \frac{2n+1}{z}$, $\lambda_n = (2n+1)^2$, $y_n(z) = 0.5 (z^2 +
(n+1/2)^2 \pi^2)$, and $p_{IG}$ and $p_\mcE$ are the corresponding densities.
The table below shows the first several probabilities for the worst case envelope, $z \simeq
1.378$.  Clearly $\bbP(L>n)$ decays rapidly with $n$.
\begin{center}
\begin{footnotesize}
\begin{tabular}{r cccc }
$n$           & 1           & 2           & 3           & 4 \\
\midrule
$\bbP(L>n)$   & $8.023 \times 10^{-4}$   & $1.728 \times 10^{-9}$   & $8.213 \times 10^{-18}$   & $8.066 \times 10^{-29}$ 
\end{tabular}
\end{footnotesize}
\end{center}
Together with Proposition 2, this provides a strong guarantee of the efficiency of the PG(1,z) sampler.

\subsection{The general PG(b, z) case}

To sample from the entire family of $\PG(b,z)$ distributions, we exploit the
additivity of the \Polya-Gamma class.  In particular, when $b \in \bbN$, one may
sample $\PG(b,z)$ by taking $b$ i.i.d.\ draws from $\PG(1,z)$ and summing them. In binomial
logistic regression, one will always sample $\PG(b,z)$ using integral $b$.  This will also be the case in negative-binomial regression if one chooses an integer over-dispersion parameter.  In the technical supplement, we discuss the case of non-integral $b$.

The run-time of the latent-variable sampling step is therefore roughly linear in the number of total counts in the data set.  For example, to sample 1 million \Polya-Gamma(1,1) random variables took $0.70$ seconds on a dual-core Apple laptop, versus $0.17$ seconds for the same number of Gamma random variables.  By contrast, to sample 1 million PG(10,1) random variables required 6.43 seconds, and to sample 1 million PG(100,1) random variables required 60.0 seconds.

We have had some initial success in developing a faster method to simulate from the PG(n,z) distribution that does not require summing together $n$ PG(1,z) draws, and that works for non-integer values of $n$.  This is an active subject of research, though somewhat beyond the scope of the present paper, where we use the sum-of-PG(1,z)'s method on all our benchmark examples.  A full report on the alternative simulation method for PG(n,z) may be found in \citet{windle:polson:scott:2013}. 

\section{Experiments}

We benchmarked the \Polya-Gamma method against several alternatives for logit and negative-binomial models.  Our purpose is to summarize the results presented in detail in our online technical supplement, to which we refer the interested reader. 

Our primary metrics of comparison are the effective sample size and the
effective sampling rate, defined as the effective sample size per second of
runtime.  The effective sampling rate quantifies how rapidly a Markov-chain
sampler can produce independent draws from the posterior distribution.
Following \cite{holmes:held:2006}, the effective sample size (ESS) for the $i$th
parameter in the model is
\begin{displaymath}
\label{ess-metric}
ESS_i = M / \{ 1 + 2 \sum_{j=1}^k \rho_i(j) \},
\end{displaymath}
where $M$ is the number of post-burn-in samples, and $\rho_i(j)$ is the $j$th
autocorrelation of $\beta_i$.  We use the \texttt{coda} package
\citep{coda-2006}, which fits an AR model to approximate the spectral density at
zero, to estimate each $ESS_i$. All of the benchmarks are generated using R so
that timings are comparable.  Some R code makes external calls to C.  In
particular, the \Polya-Gamma method calls a C routine to sample the \Polya-Gamma
random variates, just as R routines for sampling common distributions use
externally compiled code.  Here we report the median effective sample
size across all parameters in the model. Minimum and maximum effective sample
sizes are reported in the technical supplement.

Our numerical experiments support several conclusions.  

\paragraph{In binary logit models.}  First, the \Polya-Gamma is more efficient than all previously proposed data-augmentation schemes.  This is true both in terms of effective sample size and effective sampling rate.  Table \ref{tab:binarylogitsummary} summarizes the evidence: across 6 real and 2 simulated data sets, the \Polya-Gamma method was always more efficient than the next-best data-augmentation scheme (typically by a factor of 200\%--500\%).  This includes the approximate random-utility methods of \cite{obrien-dunson-2004} and \cite{fruhwirth-schnatter-fruhwirth-2010}, and the exact method of \cite{gramacy:polson:2012}.  \cite{fruhwirth-schnatter-fruhwirth-2010} find that their own method beats several other competitors, including the method of \cite{holmes:held:2006}.  We find this as well, and omit these timings from our comparison.  Further details can be found in Section 3 of the technical supplement.

\begin{table}
\begin{center}
\begin{small}
\caption{\label{tab:binarylogitsummary} \footnotesize Summary of experiments on real and simulated data for binary logistic regression.  ESS: the median effective sample size for an MCMC run of 10,000 samples.  ESR: the median effective sample rate, or median ESS divided by the runtime of the sampler in seconds.  AC: Australian credit data set.  GC1 and GC2: partial and full versions of the German credit data set.  Sim1 and Sim2: simulated data with orthogonal and correlated predictors, respectively.  Best RU-DA: the result of the best random-utility data-augmentation algorithm for that data set.  Best Metropolis: the result of the Metropolis algorithm with the most efficient proposal distribution among those tested.  See the technical supplement for full details.  }
\vspace{1pc}
\begin{tabular} {r r r r r r r r r r}
& & \multicolumn{8}{c}{Data set} \\
 &  & Nodal & Diab. & Heart & AC & GC1 & GC2 & Sim1 & Sim2 \\
 \midrule
ESS & \Polya-Gamma & 4860 & 5445 & 3527 & 3840 & 5893 & 5748 & 7692 & 2612 \\
 & Best RU-DA & 1645 & 2071 & 621 & 1044 & 2227 & 2153  & 3031 & 574 \\
 & Best Metropolis & 3609 & 5245 & 1076 & 415 & 3340 & 1050  & 4115 & 1388  \\
\midrule
ESR & \Polya-Gamma & 1632 & 964 & 634 & 300 & 383 & 258  & 2010 & 300  \\
 & Best RU-DA & 887 & 382 & 187 & 69 & 129 & 85  & 1042 & 59 \\
 & Best Metropolis & 2795 & 2524 & 544 & 122 & 933 & 223 & 2862 &  537
\end{tabular}
\end{small}
\end{center}
\end{table}

Second, the \Polya-Gamma method always had a higher effective sample size than
the two default Metropolis samplers we tried.  The first was a Gaussian proposal
using Laplace's approximation.  The second was a multivariate
$t_6$ proposal using Laplace's approximation to provide the centering and
scale-matrix parameters, recommended by \citet{rossi:allenby:mcculloch:2005} and
implemented in the R package \verb|bayesm| \citep{bayesm-2012}.

On 5 of the 8 data sets, the best Metropolis algorithm did have a higher effective sampling rate than the \Polya-Gamma method, due to the difference in run times.  But this advantage depends crucially on the proposal distribution, where even small perturbations can lead to surprisingly large declines in performance.  For example, on the Australian credit data set (labeled AC in the table), the Gaussian proposal led to a median effective sampling rate of 122 samples per second.  The very similar multivariate $t_6$ proposal led to far more rejected proposals, and gave an effective sampling rate of only 2.6 samples per second.  Diagnosing such differences for a specific problem may cost the user more time than is saved by a slightly faster sampler.

\begin{table}
\begin{center}
\begin{small}
\caption{\label{tab:mixedmodelsummary} \footnotesize Summary of experiments on real and simulated data for binary logistic mixed models.  Metropolis: the result of an independence Metropolis sampler based on the Laplace approximation.  Using a $t_6$ proposal yielded equally poor results.  See the technical supplement for full details.  }
\vspace{1pc}
\begin{tabular} {r r r r r}
& & \multicolumn{3}{c}{Data set} \\
 &  & Synthetic & Polls & Xerop \\
 \midrule
ESS & \Polya-Gamma & 6976 & 9194 & 3039  \\
 & Metropolis & 3675 & 53 & 3  \\
\midrule
ESR & \Polya-Gamma & 957 & 288 & 311  \\
 & Metropolis & 929 & 0.36 & 0.01 
\end{tabular}
\end{small}
\end{center}
\end{table}  

Finally, the \Polya-Gamma method truly shines when the model has a complex prior structure.  In general, it is difficult to design good Metropolis samplers for these problems.  For example, consider a binary logit mixed model with grouped data and a random-effects structure, where the log-odds of success for observation $j$ in group $i$ are $\psi_{ij} = \alpha_i + x_{ij} \beta_i$, and where either the $\alpha_i$, the $\beta_i$, or both receive further hyperpriors.  It is not clear that a good default
Metropolis sampler is easily constructed unless there are a large number of observations per group.  Table \ref{tab:mixedmodelsummary} shows the results of na\"ively using an independence Metropolis sampler based on the Laplace approximation to the full joint posterior.  For a synthetic data set with a balanced design of 100 observations per group, the \Polya-Gamma method is slightly better.  For the two real data sets with highly unbalanced designs, it is much better.

Of course, it is certainly possible to design and tune better Metropolis--Hastings samplers for mixed models; see, for example, \cite{gamerman-1997}.  We simply point out that what works well in the simplest case need not work well in a slightly more complicated case.  The advantages of the \Polya-Gamma method are that it requires no tuning, is simple to implement, is uniformly ergodic \citep{choi:hobert:2013}, and gives optimal or near-optimal performance across a range of cases.

\paragraph{In negative-binomial models.}

\begin{table}
\begin{center}
\begin{small}
\caption{\label{tab:nbmodelsummary} \footnotesize Summary of experiments on simulated data for negative-binomial models. Metropolis: the result of an independence Metropolis sampler based on a $t_6$ proposal. FS09: the algorithm of \cite{fruhwirth-schnatter-etal-2009}.  Sim1 and Sim2: simulated negative-binomial regression problems.  GP1 and GP2: simulated Gaussian-process spatial models.  The independence Metropolis algorithm is not applicable in the spatial models, where there as many parameters as observations.}
\vspace{1pc}
\begin{tabular} {r r r r r r}
& & \multicolumn{4}{c}{Data set} \\
 &  & Sim1 & Sim2 & GP1 & GP2 \\
 & Total Counts & 3244 & 9593 & 9137 & 22732 \\
 \midrule
ESS & \Polya-Gamma & 7646 & 3590 & 6309 & 6386  \\
	& FS09			& 719 & 915 & 1296 & 1157 \\
 & Metropolis & 749 & 764 & --- & ---   \\
\midrule
ESR & \Polya-Gamma & 285 & 52 & 62 & 3.16 \\
	& FS09			& 86 & 110 & 24 & 0.62 \\
 & Metropolis 			& 73 & 87 & --- & --- 
\end{tabular}
\end{small}
\end{center}
\end{table}  

The \Polya-Gamma method consistently yields the best effective sample sizes in
negative-binomial regression.  However, its effective sampling rate suffers when
working with a large counts or a non-integral over-dispersion
parameter.  Currently, our \Polya-Gamma sampler can draw from $\PG(b, \psi)$
quickly when $b=1$, but not for general, integral $b$: to sample from $\PG(b,
\psi)$ when $b \in \bbN$, we take $b$ independent samples of $\PG(1, \psi)$ and
sum them.  Thus in negative-binomial models, one must sample at least
$\sum_{i=1}^N y_i$ \Polya-Gamma random variates, where $y_i$ is the $i$th
response, at every MCMC iteration.  When the number of counts is relatively
high, this becomes a burden.  (The sampling method described in \citet{windle:polson:scott:2013} leads to better performance, but describing the alternative method is beyond the subject of this paper.)

The columns labeled Sim1 and Sim2 of Table \ref{tab:nbmodelsummary} show results
for data simulated from a negative-binomial model with 400 observations and 3
regressors.  (See the technical supplement for details.)  In the first case
(Sim1), the intercept is chosen so that the average outcome is a count of 8
(3244 total counts). Given the small average count size, the \Polya-Gamma method
has a superior effective sampling rate compared to the approximate method of
\cite{fruhwirth-schnatter-etal-2009}, the next-best choice.  In the second case
(Sim2), the average outcome is a count of 24 (9593 total counts).  Here the
\citeauthor{fruhwirth-schnatter-etal-2009} algorithm finishes more quickly, and
therefore has a better effective sampling rate.  In both cases we restrict the sampler to integer over-dispersion parameters.

As before, the \Polya-Gamma method starts to shine when working with more complicated
hierarchical models that devote proportionally less time to sampling the auxiliary variables.
For instance, consider a spatial model where we observe counts $y_1, \ldots, y_n$ at locations $x_1, \ldots, x_n$, respectively.  It is natural to model the log rate parameter as a Gaussian process:
$$
y_i \sim NB(n, 1/\{1+e^{-\psi_i}\}) \; , \quad
\psi  \sim GP(0, K) \, ,
$$
where $\psi = (\psi_1, \ldots, \psi_n)^T$ and $K$ is constructed by evaluating a covariance kernel at the locations $x_i$.  For example, under the squared-exponential kernel, we have
\[
K_{ij} = \kappa + \exp \left\{ \frac{d(x_i, x_j)^2}{2 \ell^2} \right\},
\]
with characteristic length scale $\ell$, nugget $\kappa$, and distance function $d$ (in our examples, Euclidean distance).

Using either the \Polya-Gamma or the \cite{fruhwirth-schnatter-etal-2009} techniques, one arrives at a multivariate Gaussian conditional for $\psi$ whose covariance matrix involves latent variables.  Producing a random variate from this distribution is expensive, as one must calculate the Cholesky decomposition
of a relatively large matrix at each iteration.  Therefore, the overall sampler spends relatively less time drawing auxiliary variables.  Since the \Polya-Gamma method leads to a higher effective sample size, it wastes fewer of the expensive draws for the main parameter.

The columns labeled GP1 and GP2 of Table \ref{tab:nbmodelsummary} show two such examples.  In the first synthetic data set, 256 equally spaced $x$ points were used to generate a draw for $\psi$ from a Gaussian process with length scale $\ell = 0.1$ and nugget $\kappa=0.0$.  The average count was $\bar y = 35.7$, or 9137 total counts (roughly the same as in the second regression example, Sim2).  In the second synthetic data set,  we simulated $\psi$ from a Gaussian process over 1000 $x$ points, with length scale $\ell = 0.1$ and a nugget $=0.0001$.  This yielded 22,720 total counts.  In both cases, the \Polya-Gamma method led to a more efficient sampler---by a factor of 3 for the smaller problem, and 5 for the larger.

\section{Discussion}
\label{sec:extensions}

We have shown that Bayesian inference for logistic models can be implemented using a data augmentation scheme
based on the novel class of \Polya-Gamma distributions. This leads to simple Gibbs-sampling algorithms for posterior computation that exploit standard normal linear-model theory, and that are notably simpler than previous schemes.  We have also constructed an accept/reject sampler for the new family, with strong guarantees of efficiency (Propositions 2 and 3).

The evidence suggests that our data-augmentation scheme is the best current method for fitting complex Bayesian hierarchical models with binomial likelihoods.  It also opens the door for exact Bayesian treatments of many modern-day machine-learning classification methods based on mixtures of logits \citep[e.g.][]{salakhutdinov:etal:2007,blei:lafferty:2007}.  Applying the \Polya-Gamma mixture framework to such problems is currently an active area of research.

Moreover, posterior updating via exponential tilting is a quite general situation that arises in Bayesian inference incorporating latent variables.  In our case, the posterior distribution of $ \omega $ that arises under normal pseudo-data with precision $\omega$ and a $\PG(b,0)$ prior is precisely an exponentially titled $\PG(b,0)$ random variable.  This led to our characterization of the general $\PG(b, c)$ class.   An interesting fact is that we were able to identify the conditional posterior for the latent variable 
strictly using its moment-generating function, without ever appealing to Bayes' rule for
density functions.  This follows the L\'evy-penalty framework of \citet{Polson:Scott:2010b}
and relates to work by \citet{ciesielski:taylor:1962} on the sojourn times of Brownian motion.  There may be many other situations where the same idea is applicable.

Our benchmarks have relied upon serial computation.  However,
one may trivially parallelize a vectorized \Polya-Gamma draw on a multicore CPU.
Devising such a sampler for a graphical-processing unit (GPU) is less straightforward, but potentially more
fruitful. The massively parallel nature of GPUs offer a solution to the
sluggishness found when sampling $\PG(n,z)$ variables for large, integral $n$,
which was the largest source of inefficiency with the negative-binomial results presented earlier.

\paragraph{Acknowledgements.} The authors wish to thank Hee Min Choi and Jim Hobert for sharing an early draft of their paper on the uniform ergodicity of the \Polya-Gamma Gibbs sampler.  They also wish to thank two anonymous referees, the associate editor, and the editor of the \textit{Journal of the American Statistical Association}, whose many insights and helpful suggestions have improved the paper.  The second author acknowledges the support of a CAREER grant from the U.S.~National Science Foundation (DMS-1255187).

\end{spacing}

\singlespace

\bibliographystyle{abbrvnat}
\bibliography{masterbib,rpackages}

\newpage

\appendix

\renewcommand{\thesection}{S\arabic{section}}

{ \huge \noindent Technical Supplement}

\section{Details of \Polya-Gamma sampling algorithm}

\begin{figure}
\centering
\includegraphics[scale=0.5]{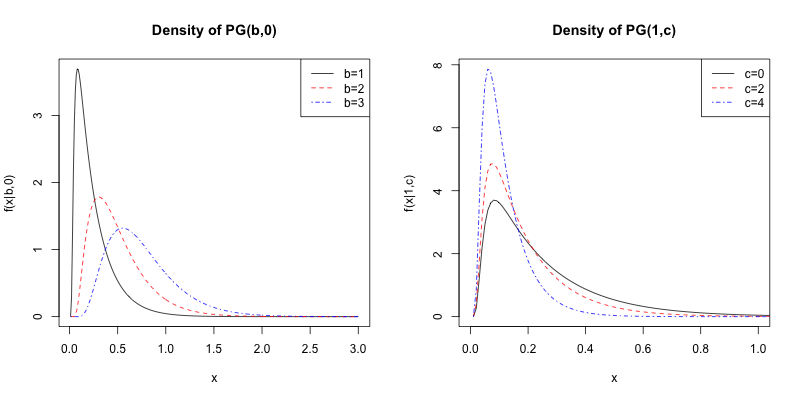}
\caption{\label{fig:pg-dens} Plots of the density of the \Polya-Gamma
  distribution $\PG(b,c)$ for various values of $b$ and $c$.  Note that the
  horizontal and vertical axes differ in each plot.}
\end{figure}


\begin{algorithm}[p]
\caption{\label{alg:tiltedpg} Sampling from $PG(1,z)$}
\begin{algorithmic}
\State \textbf{Input}: $z > 0$.

\State \textbf{Define}: \texttt{pigauss($t \mid \mu, \lambda$)}, the CDF of the
inverse Gaussian distribution

\State \textbf{Define}: $a_n(x)$, the piecewise-defined coefficients in (1) and (2).

\State $z \leftarrow |z| / 2$,
       $t \leftarrow 0.64$,
       $K \leftarrow \pi^2 / 8 + z^2 / 2$
\State $p \leftarrow \frac{\pi}{2 K} \exp(-K t)$

\State $q \leftarrow 2 \exp(-|z|) \; \texttt{pigauss}(t \mid \mu = 1/z, \lambda
= 1.0)$


\Repeat
\State Generate $U, V \sim \mathcal{U}(0,1)$
\If{$U < p / (p+q)$}
\State (Truncated Exponential)
\State $X \leftarrow t + E/K$ where $E \sim \mathcal{E}(1)$
\Else
\State (Truncated Inverse Gaussian)
\State $\mu \leftarrow 1/z$
\If {$\mu > t$}
\Repeat
\State Generate $1/X \sim \chi^2_1 \1_{(t,\infty)}$
\Until {$\mathcal{U}(0,1) < \exp(-\frac{z^2}{2} X)$}
\Else
\Repeat
\State Generate $X \sim \mathcal{IN}(\mu, 1.0)$
\Until {$X < t$}
\EndIf
\EndIf
\State $S \leftarrow a_0(X)$, $Y \leftarrow VS$, $n \leftarrow 0$
\Repeat
\State $n \leftarrow n + 1$
\If{$n$ is odd}
\State $S \leftarrow S - a_n(X)$; \textbf{if} $Y < S$, \textbf{then} \Return $X$
/ 4
\Else
\State $S \leftarrow S + a_n(X)$; \textbf{if} $Y > S$, \textbf{then} \textbf{break}
\EndIf
\Until {FALSE}
\Until {FALSE}
\end{algorithmic}
\end{algorithm}


\begin{algorithm}
\begin{algorithmic}
\State \textbf{Input}: $\mu, t > 0$.

\State Let $z = 1/\mu$.
\Repeat 
\Repeat
\State Generate $E, E' \sim \mcE(1)$.
\Until $E^2 \leq 2 E' / t$
\State $X \leftarrow t / (1 + tE)^2$
\State $\alpha \leftarrow \exp(\frac{-1}{2} z^2 X)$
\State $U \sim \mcU$
\Until $U \leq \alpha$
\end{algorithmic}
\caption{Algorithm used to generate $IG(\mu, 1) \1_{(0,t)}$ when
  $\mu > t$.}
\label{alg:ig1}
\end{algorithm}


\begin{algorithm}
\begin{algorithmic}
\State \textbf{Input}: $\mu, t > 0$.
\Repeat
\State $Y \sim N(0,1)^2$.
\State $X \leftarrow \mu + 0.5 \mu^2 Y - 0.5 \mu \sqrt{4 \mu Y + (\mu Y)^2}$
\State $U \sim \mcU$
\State \textbf{If} $(U > \mu / (\mu + X))$, then $X \leftarrow \mu^2 / X$.
\Until $X \leq R$.
\end{algorithmic}
\caption{Algorithm used to generate $IG(\mu, 1)  \1_{(0,t)}$ when
  $\mu \leq t$.}
\label{alg:ig2}
\end{algorithm}


Algorithm \ref{alg:tiltedpg} shows pseudo-code for sampling the \Polya-Gamma$(1,z)$ distribution.  Recall from the main manuscript that one may pick $t > 0$ and define the piecewise coefficients
\begin{numcases}{a_n(x) =}
\label{jacobi:dL}
\pi (n+1/2) \; \left(\frac{2}{\pi x}\right)^{3/2}
  \exp \left\{ - \frac{2(n+1/2)^2}{x} \right\} & $0 < x \leq t$, \\
\label{jacobi:dR}
\pi (n+1/2) \;
  \exp \left\{ -\frac{(n+1/2)^2 \pi^2}{2} x \right\} & $x > t$,
\end{numcases}
so that $f(x) = \sum_{n=0}^\infty (-1)^n a_n(x)$ satisfies the partial sum
criterion for $x > 0$.

To complete the analysis of the \Polya-Gamma sampler, we specify our method for
sampling truncated inverse Gaussian random variables, $IG(1/z,1) \bbI_{(0,t]}$.
When $z$ is small the inverse Gaussian distribution is approximately inverse
$\chi^2_1$, motivating an accept-reject algorithm.  When $z$ is large, most of
the inverse Gaussian distribution's mass will be below the truncation point $t$,
motivating a rejection algorithm.  Thus, we take a two pronged approach.

When $1/z > t$ we generate a truncated inverse-Gaussian random variate using
accept-reject sampling using the proposal distribution $ (1/\chi^2_1)
\bbI_{(t,\infty)}$.  The proposal $X$ is generated following
\citet{devroye:2009}.
Considering the ratio of the kernels, one finds that $P(\text{accept} | X = x) =
\exp ( - x z^2 / 2 )$.  
Since $z < 1/t$ and $X < t$ we may compute a lower bound
on the average rate of acceptance:
\[
\bbE \Big[ \exp \Big( \frac{-z^2}{2} X \Big) \Big] \geq \exp \frac{-1}{2t} =
0.61 \, .
\]
See algorithm (\ref{alg:ig1}) for pseudocode.

When $1/z \leq t$, we generate a truncated inverse-Gaussian random variate using
rejection sampling.  \citet{devroye:1986} (p.~149) describes how to sample from
an inverse-Gaussian distribution using a many-to-one transformation.  Sampling
$X$ in this fashion until $X < t$ yields an acceptance rate bounded below by
\[
\int_0^t IG(x | 1/z, \lambda=1) dx \geq \int_0^t IG(x | t, \lambda=1) = 0.67
\]
for all $1/z < t$.  See Algorithm \ref{alg:ig2} for pseudocode.

Recall that when $b$ is an integer, we draw $\PG(b,z)$ by summing $b$ i.i.d.~draws from $\PG(1,z)$.  When $b$ is not integral, the following simple approach often suffices.  Write $b = \floor{b} + e$, where $\floor{b}$ is the
integral part of $b$, and sum a draw from $\PG(\floor{b}, z)$, using the method
previously described, with a draw from $\PG(e, z)$, using the finite
sum-of-gammas approximation.  With 200 terms in the sum, we find that the
approximation is quite accurate for such small values of the first parameter, as
each $\mbox{Ga}(e,1)$ term in the sum tends to be small, and the weights in the
sum decay like $1/k^2$.  This, in contrast, may not be the case when using the
finite sum-of-gammas approximation for arbitrary $b$.

In \citet{windle:polson:scott:2013}, we describe a better method for handling large and/or non-integer shape parameters.  This method is implemented in the \verb|BayesLogit| R package \citep{bayeslogit:2013}.

\section{Benchmarks: overview}

We benchmark the \Polya-Gamma method against several alternatives for binary
logistic regression and negative binomial regression for count data to measure
its relative performance.  All of these benchmarks are empirical and hence some
caution is urged.  Our primary metric of comparison is the effective sampling
rate, which is the effective sample size per second and which quantifies how
quickly a sampler can produce independent draws from the posterior distribution.
However, this metric is sensitive to numerous idiosyncrasies relating to the
implementation of the routines, the language in which they are written, and the
hardware on which they are run.  We generate these benchmarks using R, though
some of the routines make calls to external C code.  The specifics of each
method are discussed in further detail below.  In general, we find that the
\Polya-Gamma technique compares favorably to other data augmentation methods.  Specifically, the \Polya-Gamma technique performs better than the methods of
\cite{obrien-dunson-2004}, \cite{gramacy:polson:2012}, and
\cite{fruhwirth-schnatter-fruhwirth-2010}. \cite{fruhwirth-schnatter-fruhwirth-2010}
provides a detailed comparison of several methods itself.  For instance, the
authors find that method of \cite{holmes:held:2006} did not beat their discrete
mixture of normals.  We find this as well and hence omit it from the
comparisons below.

For each data set, we run 10 MCMC simulations with 12,000 samples each,
discarding the first 2,000 as burn-in, thereby leaving 10 batches of 10,000
samples.  The effective sample size for each regression coefficient is
calculated using the \texttt{coda} \citep{coda-2006} package and averaged across
the 10 batches.  The component-wise minimum, median, and maximum of the
(average) effective sample sizes are reported to summarize the results.  A
similar calculation is performed to calculate minimum, median, and maximum
effective sampling rates (ESR).  The effective sampling rate is the ratio of the
effective sample size to the time taken to produce the sample.  Thus, the
effective sampling rates are normalized by the time taken to produce the 10,000
samples, disregarding the time taken for initialization, preprocessing, and
burn-in.  When discussing the various methods the primary metric we refer to is
the median effective sampling rate, following the example of
\cite{fruhwirth-schnatter-fruhwirth-2010}.

All of these experiments are carried out using R 2.15.1 on an Ubuntu machine
with 8GB or RAM and an Intel Core i5 quad core processor.  The number of cores
is a potentially important factor as some libraries, including those that perform the matrix operations in R, may take advantage of multiple cores.  The C code that we have written does not use parallelism.

In the sections that follow, each table reports the following metrics:
\begin{compactitem}
\item the execution time of each method in seconds;
\item the acceptance rate (relevant for the Metropolis samplers);
\item the minimum, median, and maximum effective sample sizes (ESS) across all
  fixed or random effects; and
\item the minimum, median, and maximum effective sampling rates (ESR) across all
  fixed or random effects, defined as the effective sample size per second of
  runtime.
\end{compactitem}

\section{Benchmarks: binary logistic regression}

\subsection{Data Sets}
\label{sec:blogit-datasets}
\begin{compactdesc}
\item[Nodal:] part of the \texttt{boot} R package \citep{boot-2012}. The
  response indicates if cancer has spread from the prostate to surrounding lymph
  nodes.  There are 53 observations and 5 binary predictors.

\item[Pima Indian:] There are 768 observations and 8 continuous predictors.  It is noted on the UCI
website\footnote{\url{http://archive.ics.uci.edu/ml/datasets/Pima+Indians+Diabetes}} that there are many predictor values coded as 0, though the physical measurement should be non-zero.  We have removed all of those entries to
generate a data set with 392 observations.  The marginal mean incidence of diabetes is roughly 0.33 before and after removing these data points.

\item[Heart:] The response represents either an absence or presence of heart disease.\footnote{\url{http://archive.ics.uci.edu/ml/datasets/Statlog+(Heart)}}  There are 270 observations and 13 attributes, of which 6 are categorical or binary and
1 is ordinal.  The ordinal covariate has been stratified by dummy variables.

\item[Australian Credit:] The response represents either accepting or rejecting
  a credit card
  application.\footnote{\url{http://archive.ics.uci.edu/ml/datasets/Statlog+(Australian+Credit+Approval)}.}
  The meaning of each predictor was removed to protect the propriety of the
  original data.  There are 690 observations and 14 attributes, of which 8 are
  categorical or binary.  There were 37 observations with missing attribute
  values.  These missing values were replaced by the mode of the attribute in
  the case of categorical data and the mean of the attribute for continuous
  data.  This dataset is linearly separable and results in some divergent
  regression coefficients, which are kept in check by the prior.

\item[German Credit 1 and 2:] The response represents either a good or bad credit risk.\footnote{\url{http://archive.ics.uci.edu/ml/datasets/Statlog+(German+Credit+Data)}}  There are 1000
observations and 20 attributes, including both continuous and categorical data.
We benchmark two scenarios.  In the first, the ordinal covariates have been
given integer values and have not been stratified by dummy variables, yielding a
total of 24 numeric predictors.  In the second, the ordinal data has been
stratified by dummy variables, yielding a total of 48 predictors.

\item[Synthetic 1:] Simulated data with 150 outcomes and 10 predictors.  The design points were chosen to be orthogonal.  The data are included as a supplemental file.

\item[Synthetic 2:] Simulated data with 500 outcomes and 20 predictors.  The design points were simulated from a Gaussian factor model, to yield pronounced patterns of collinearity.  The data are included as a supplemental file.

\end{compactdesc}

\subsection{Methods}

All of these routines are implemented in R, though some of them make calls to C.
In particular, the independence Metropolis samplers do not make use of any
non-standard calls to C, though their implementations have very little R
overhead in terms of function calls.  The \Polya-Gamma method calls a C routine
to sample the \Polya-Gamma random variates, but otherwise only uses R.

As a check upon our independence Metropolis sampler we include the
independence Metropolis sampler of \cite{rossi:allenby:mcculloch:2005}, which may be found in
the \texttt{bayesm} package \citep{bayesm-2012}.  Their sampler uses a
$t_6$ proposal, while ours uses a normal proposal.  The suite of routines in the
\texttt{binomlogit} package \citep{binomlogit-2012} implement the techniques
discussed in \cite{fussl-etal-2011-slides}.  One routine provided by the
\texttt{binomlogit} package coincides with the technique described in
\cite{fruhwirth-schnatter-fruhwirth-2010} for the case of binary logistic
regression.  A separate routine implements the latter and uses a single call to
C.  Gramacy and Polson's R package, \texttt{reglogit}, also calls external C
code \citep{reglogit-2012}.

For every data set the regression coefficient was
given a diffuse $N(0, 0.01 I)$ prior, except when using Gramacy and Polson's
method, in which case it was given a $\exp(\sum_{i} |\beta_i / 100|)$ prior per
the specifications of the \texttt{reglogit} package.  The following is a short
description of each method along with its abbreviated name.

\begin{compactdesc}

\item PG: The \Polya-Gamma method described previously.

\item FS: \cite{fruhwirth-schnatter-fruhwirth-2010} follow \cite{holmes:held:2006}
and use the representation
\begin{equation}
\label{eqn:hh-rep}
y_i = \1 \{ z_i > 0 \} \;, \quad z_i = x_i \beta + \ep_i \;, \quad \ep_i \sim \text{Lo} \, ,
\end{equation}
where $\text{Lo}$ is the standard logistic distribution \citep[c.f.][for the probit case]{albert:chib:1993}.   They approximate $p(\ep_i)$ using a discrete mixture of normals.

\item IndMH: Independence Metropolis with a normal proposal using the
posterior mode and the Hessian at the mode for the mean and precision matrix.

\item RAM: after Rossi, Allenby, and McCulloch.  An independence Metropolis
  with a $t_6$ proposal from the R package \texttt{bayesm} \citep{bayesm-2012}.
  Calculate the posterior mode and the Hessian at the mode to pick the mean and
  scale matrix of the proposal.

\item OD: The method of \cite{obrien-dunson-2004}.  Strictly speaking, this is not
logistic regression; it is binary regression using a Student-$t$ cumulative
distribution function as the inverse link function.

\item dRUMAuxMix: Work by \cite{fussl-etal-2011-slides} that extends the
  technique of \cite{fruhwirth-schnatter-fruhwirth-2010}.  A convenient
  representation is found that relies on a discrete mixture of normals
  approximation for posterior inference that works for binomial logistic
  regression.  From the R package \texttt{binomlogit} \citep{binomlogit-2012}.

\item dRUMIndMH: Similar to dRUMAuxMix, but instead of using a discrete mixture of
normals, use a single normal to approximate the error term and correct using
Metropolis-Hastings.  From the R package \texttt{binomlogit}.

\item IndivdRUMIndMH: This is the same as dRUMIndMH, but specific to binary
logistic regression.  From the R package \texttt{binomlogit}.

\item dRUMHAM: Identical to dRUMAuxMix, but now use a discrete mixture of normals
approximation in which the number of components to mix over is determined by
$y_i / n_i$.  From the R package \texttt{binomlogit}.

\item GP: after \cite{gramacy:polson:2012}.  Another data augmentation scheme
  with only a single layer of latents.  This routine uses a double exponential
  prior, which is hard-coded in the R package \texttt{reglogit}
  \citep{reglogit-2012}.  We set the scale of this prior to agree with the scale
  of the normal prior we used in all other cases above.

\end{compactdesc}

\subsection{Results}

The results are shown in Tables \ref{tab:blogit-nodal} through \ref{tab:blogit-synth2}.  As mentioned previously, these are averaged over 10 runs.

\begin{table}
\tablesize
\centering
\caption{\label{tab:blogit-nodal} Nodal data: $N= 53$, $P=6$}
  \vspace{\baselineskip}
\begin{tabular}{l r r r r r r r r } 
\toprule
          Method  &     time &    ARate &  ESS.min &  ESS.med &  ESS.max &  ESR.min &  ESR.med &  ESR.max \\  
          \midrule
              PG  &     2.98 &     1.00 &  3221.12 &  4859.89 &  5571.76 &  1081.55 &  1631.96 &  1871.00 \\ 
           IndMH  &     1.76 &     0.66 &  1070.23 &  1401.89 &  1799.02 &   610.19 &   794.93 &  1024.56 \\ 
             RAM  &     1.29 &     0.64 &  3127.79 &  3609.31 &  3993.75 &  2422.49 &  2794.69 &  3090.05 \\ 
              OD  &     3.95 &     1.00 &   975.36 &  1644.66 &  1868.93 &   246.58 &   415.80 &   472.48 \\ 
              FS  &     3.49 &     1.00 &   979.56 &  1575.06 &  1902.24 &   280.38 &   450.67 &   544.38 \\ 
      dRUMAuxMix  &     2.69 &     1.00 &  1015.18 &  1613.45 &  1912.78 &   376.98 &   598.94 &   710.30 \\ 
       dRUMIndMH  &     1.41 &     0.62 &   693.34 &  1058.95 &  1330.14 &   492.45 &   751.28 &   943.66 \\ 
  IndivdRUMIndMH  &     1.30 &     0.61 &   671.76 &  1148.61 &  1339.58 &   518.79 &   886.78 &  1034.49 \\ 
         dRUMHAM  &     3.06 &     1.00 &   968.41 &  1563.88 &  1903.00 &   316.82 &   511.63 &   622.75 \\ 
              GP  &    17.86 &     1.00 &  2821.49 &  4419.37 &  5395.29 &   157.93 &   247.38 &   302.00 \\
              \bottomrule
 \end{tabular}
 \end{table}

 \begin{table}
 \tablesize
 \centering
 \label{tab:blogit-diabetes}
  \caption{Diabetes data, N=270, P=19}
  \vspace{\baselineskip}
\begin{tabular}{l r r r r r r r r } 
\toprule
          Method  &     time &    ARate &  ESS.min &  ESS.med &  ESS.max &  ESR.min &  ESR.med &  ESR.max \\ 
          \midrule
              PG  &     5.65 &     1.00 &  3255.25 &  5444.79 &  6437.16 &   576.14 &   963.65 &  1139.24 \\ 
           IndMH  &     2.21 &     0.81 &  3890.09 &  5245.16 &  5672.83 &  1759.54 &  2371.27 &  2562.59 \\ 
             RAM  &     1.93 &     0.68 &  4751.95 &  4881.63 &  5072.02 &  2456.33 &  2523.85 &  2621.98 \\ 
              OD  &     6.63 &     1.00 &  1188.00 &  2070.56 &  2541.70 &   179.27 &   312.39 &   383.49 \\ 
              FS  &     6.61 &     1.00 &  1087.40 &  1969.22 &  2428.81 &   164.39 &   297.72 &   367.18 \\ 
      dRUMAuxMix  &     6.05 &     1.00 &  1158.42 &  1998.06 &  2445.66 &   191.52 &   330.39 &   404.34 \\ 
       dRUMIndMH  &     3.82 &     0.49 &   647.20 &  1138.03 &  1338.73 &   169.41 &   297.98 &   350.43 \\ 
  IndivdRUMIndMH  &     2.91 &     0.48 &   614.57 &  1111.60 &  1281.51 &   211.33 &   382.23 &   440.63 \\ 
         dRUMHAM  &     6.98 &     1.00 &  1101.71 &  1953.60 &  2366.54 &   157.89 &   280.01 &   339.18 \\ 
              GP  &    88.11 &     1.00 &  2926.17 &  5075.60 &  5847.59 &    33.21 &    57.61 &    66.37 \\
              \bottomrule
 \end{tabular}
  \end{table}

 \begin{table}
 \tablesize
 \centering
 \label{tab:blogit-heart}
   \caption{Heart data: $N = 270$, $P=19$}
  \vspace{\baselineskip}
\begin{tabular}{l r r r r r r r r } 
\toprule
          Method  &     time &    ARate &  ESS.min &  ESS.med &  ESS.max &  ESR.min &  ESR.med &  ESR.max \\  
          \midrule
              PG  &     5.56 &     1.00 &  2097.03 &  3526.82 &  4852.37 &   377.08 &   633.92 &   872.30 \\ 
           IndMH  &     2.24 &     0.39 &   589.64 &   744.86 &   920.85 &   263.63 &   333.19 &   413.03 \\ 
             RAM  &     1.98 &     0.30 &   862.60 &  1076.04 &  1275.22 &   436.51 &   543.95 &   645.13 \\ 
              OD  &     6.68 &     1.00 &   620.90 &  1094.27 &  1596.40 &    93.03 &   163.91 &   239.12 \\ 
              FS  &     6.50 &     1.00 &   558.95 &  1112.53 &  1573.88 &    85.92 &   171.04 &   241.96 \\ 
      dRUMAuxMix  &     5.97 &     1.00 &   604.60 &  1118.89 &  1523.84 &   101.33 &   187.49 &   255.38 \\ 
       dRUMIndMH  &     3.51 &     0.34 &   256.85 &   445.87 &   653.13 &    73.24 &   127.28 &   186.38 \\ 
  IndivdRUMIndMH  &     2.88 &     0.35 &   290.41 &   467.93 &   607.80 &   100.70 &   162.25 &   210.79 \\ 
         dRUMHAM  &     7.06 &     1.00 &   592.63 &  1133.59 &  1518.72 &    83.99 &   160.72 &   215.25 \\ 
              GP  &    65.53 &     1.00 &  1398.43 &  2807.09 &  4287.55 &    21.34 &    42.84 &    65.43 \\
              \bottomrule
 \end{tabular}
\end{table}

\begin{table}
\tablesize
\centering
 \label{tab:blogit-aus}
    \caption{Australian Credit: $N = 690$, $P=35$}
  \vspace{\baselineskip}
\begin{tabular}{l r r r r r r r r } 
\toprule
          Method  &     time &    ARate &  ESS.min &  ESS.med &  ESS.max &  ESR.min &  ESR.med &  ESR.max \\  
          \midrule
              PG  &    12.78 &     1.00 &   409.98 &  3841.02 &  5235.53 &    32.07 &   300.44 &   409.48 \\ 
           IndMH  &     3.42 &     0.22 &   211.48 &   414.87 &   480.02 &    61.89 &   121.53 &   140.59 \\ 
             RAM  &     3.92 &     0.00 &     8.27 &    10.08 &    26.95 &     2.11 &     2.57 &     6.87 \\ 
              OD  &    14.59 &     1.00 &    28.59 &   988.30 &  1784.77 &     1.96 &    67.73 &   122.33 \\ 
              FS  &    15.05 &     1.00 &    36.22 &  1043.69 &  1768.47 &     2.41 &    69.37 &   117.53 \\ 
      dRUMAuxMix  &    14.92 &     1.00 &    29.34 &   991.32 &  1764.40 &     1.97 &    66.44 &   118.27 \\ 
       dRUMIndMH  &     8.93 &     0.19 &    13.03 &   222.92 &   435.42 &     1.46 &    24.97 &    48.76 \\ 
  IndivdRUMIndMH  &     7.38 &     0.19 &    13.61 &   220.02 &   448.76 &     1.85 &    29.83 &    60.84 \\ 
         dRUMHAM  &    18.64 &     1.00 &    28.75 &  1040.74 &  1817.85 &     1.54 &    55.84 &    97.53 \\ 
              GP  &   162.73 &     1.00 &    95.81 &  2632.74 &  4757.04 &     0.59 &    16.18 &    29.23 \\
              \bottomrule
 \end{tabular}
 \end{table}

 \begin{table}
 \tablesize
 \centering
 \label{tab:blogit-ger.num}
     \caption{German Credit 1: $N = 1000$, $P=25$}
  \vspace{\baselineskip}
\begin{tabular}{l r r r r r r r r } 
\toprule
          Method  &     time &    ARate &  ESS.min &  ESS.med &  ESS.max &  ESR.min &  ESR.med &  ESR.max \\ 
          \midrule
              PG  &    15.37 &     1.00 &  3111.71 &  5893.15 &  6462.36 &   202.45 &   383.40 &   420.44 \\ 
           IndMH  &     3.58 &     0.68 &  2332.25 &  3340.54 &  3850.71 &   651.41 &   932.96 &  1075.47 \\ 
             RAM  &     4.17 &     0.43 &  1906.23 &  2348.20 &  2478.68 &   457.11 &   563.07 &   594.30 \\ 
              OD  &    17.32 &     1.00 &  1030.53 &  2226.92 &  2637.98 &    59.51 &   128.59 &   152.33 \\ 
              FS  &    18.21 &     1.00 &   957.05 &  2154.06 &  2503.09 &    52.55 &   118.27 &   137.43 \\ 
      dRUMAuxMix  &    18.13 &     1.00 &   955.41 &  2150.59 &  2533.40 &    52.68 &   118.60 &   139.70 \\ 
       dRUMIndMH  &    10.60 &     0.29 &   360.72 &   702.89 &   809.20 &    34.03 &    66.30 &    76.33 \\ 
  IndivdRUMIndMH  &     8.35 &     0.29 &   334.83 &   693.41 &   802.33 &    40.09 &    83.04 &    96.08 \\ 
         dRUMHAM  &    22.15 &     1.00 &   958.02 &  2137.13 &  2477.10 &    43.25 &    96.48 &   111.84 \\ 
              GP  &   223.80 &     1.00 &  2588.07 &  5317.57 &  6059.81 &    11.56 &    23.76 &    27.08 \\
              \bottomrule
 \end{tabular}
 \end{table}

 \begin{table}
 \tablesize
 \centering
 \label{tab:blogit-ger}
      \caption{German Credit 2: $N = 1000$, $P=49$}
  \vspace{\baselineskip}
\begin{tabular}{l r r r r r r r r } 
\toprule
          Method  &     time &    ARate &  ESS.min &  ESS.med &  ESS.max &  ESR.min &  ESR.med &  ESR.max \\ 
          \midrule
              PG  &    22.30 &     1.00 &  2803.23 &  5748.30 &  6774.82 &   125.69 &   257.75 &   303.76 \\ 
           IndMH  &     4.72 &     0.41 &   730.34 &  1050.29 &  1236.55 &   154.73 &   222.70 &   262.05 \\ 
             RAM  &     6.02 &     0.00 &     5.49 &    14.40 &   235.50 &     0.91 &     2.39 &    39.13 \\ 
              OD  &    25.34 &     1.00 &   717.94 &  2153.05 &  2655.86 &    28.33 &    84.96 &   104.80 \\ 
              FS  &    26.44 &     1.00 &   727.17 &  2083.48 &  2554.62 &    27.50 &    78.80 &    96.62 \\ 
      dRUMAuxMix  &    26.91 &     1.00 &   755.31 &  2093.68 &  2562.11 &    28.06 &    77.80 &    95.21 \\ 
       dRUMIndMH  &    14.66 &     0.13 &   132.74 &   291.11 &   345.12 &     9.05 &    19.86 &    23.54 \\ 
  IndivdRUMIndMH  &    12.45 &     0.13 &   136.57 &   290.13 &   345.22 &    10.97 &    23.31 &    27.73 \\ 
         dRUMHAM  &    35.99 &     1.00 &   742.04 &  2075.41 &  2579.42 &    20.62 &    57.67 &    71.67 \\ 
              GP  &   243.41 &     1.00 &  2181.84 &  5353.41 &  6315.71 &     8.96 &    21.99 &    25.95 \\
              \bottomrule
 \end{tabular}
\end{table}


\begin{table}
\tablesize
\centering
      \caption{Synthetic 1, orthogonal predictors: $N=150$, $P=10$}
  \vspace{\baselineskip}
\begin{tabular}{l r r r r r r r r } 
\toprule
          Method  &   time &  ARate & ESS.min & ESS.med & ESS.max & ESR.min & ESR.med & ESR.max \\ 
          \midrule
            PG  &     3.83 &     1.00 &   6140.81 &   7692.04 &   8425.59 &   1604.93 &   2010.44 &   2201.04 \\ 
            FS  &     4.46 &     1.00 &   2162.42 &   2891.85 &   3359.98 &    484.91 &    648.41 &    753.38 \\ 
         IndMH  &     1.87 &     0.78 &   3009.10 &   4114.86 &   4489.16 &   1609.67 &   2200.72 &   2397.94 \\ 
           RAM  &     1.54 &     0.64 &   3969.87 &   4403.51 &   4554.04 &   2579.84 &   2862.12 &   2960.05 \\ 
            OD  &     4.88 &     1.00 &   2325.65 &   3030.71 &   3590.09 &    476.36 &    620.74 &    735.29 \\ 
     dRUMIndMH  &     2.10 &     0.53 &   1418.07 &   1791.71 &   2030.70 &    676.70 &    854.94 &    968.96 \\ 
       dRUMHAM  &     4.34 &     1.00 &   2170.71 &   2887.57 &   3364.68 &    500.67 &    666.18 &    776.37 \\ 
    dRUMAuxMix  &     3.79 &     1.00 &   2207.30 &   2932.21 &   3318.37 &    583.11 &    774.58 &    876.59 \\ 
IndivdRUMIndMH  &     1.72 &     0.53 &   1386.35 &   1793.50 &   2022.31 &    805.40 &   1042.20 &   1174.97 \\ 
            GP  &    38.53 &     1.00 &   5581.31 &   7284.98 &   8257.91 &    144.85 &    189.07 &    214.32 \\
            \bottomrule
 \end{tabular}
\end{table}


\begin{table}
\tablesize
\centering
      \caption{ \label{tab:blogit-synth2} Synthetic 2, correlated predictors: $N=500$, $P=20$}
  \vspace{\baselineskip}
\begin{tabular}{l r r r r r r r r } 
\toprule
          Method  &   time &  ARate & ESS.min & ESS.med & ESS.max & ESR.min & ESR.med & ESR.max \\ 
          \midrule
            PG  &     8.70 &     1.00 &   1971.61 &   2612.10 &   2837.41 &    226.46 &    300.10 &    325.95 \\ 
            FS  &     9.85 &     1.00 &    459.59 &    585.91 &    651.05 &     46.65 &     59.48 &     66.09 \\ 
         IndMH  &     2.52 &     0.42 &    826.94 &    966.95 &   1119.81 &    327.98 &    382.96 &    443.65 \\ 
           RAM  &     2.59 &     0.34 &   1312.67 &   1387.94 &   1520.29 &    507.54 &    536.84 &    588.10 \\ 
            OD  &     9.67 &     1.00 &    428.12 &    573.75 &    652.30 &     44.28 &     59.36 &     67.48 \\ 
     dRUMIndMH  &     5.35 &     0.33 &    211.14 &    249.33 &    281.50 &     39.46 &     46.58 &     52.59 \\ 
       dRUMHAM  &    11.18 &     1.00 &    452.50 &    563.30 &    644.73 &     40.46 &     50.37 &     57.65 \\ 
    dRUMAuxMix  &     9.51 &     1.00 &    422.00 &    564.95 &    639.89 &     44.39 &     59.43 &     67.31 \\ 
IndivdRUMIndMH  &     4.17 &     0.32 &    201.50 &    239.50 &    280.35 &     48.37 &     57.51 &     67.30 \\ 
            GP  &   114.98 &     1.00 &    748.71 &   1102.59 &   1386.08 &      6.51 &      9.59 &     12.06 \\
            \bottomrule
 \end{tabular}

\end{table}


\section{Benchmarks: logit mixed models}

A major advantage of data augmentation, and hence the \Polya-Gamma
technique, is that it is easily adapted to more complicated models.  We consider three examples of logistic mixed model whose intercepts are random effects, in which case the log odds for observation $j$ from group $i$, $\psi_{ij}$, is
modeled by
\begin{eqnarray}
\psi_{ij} &=& \alpha_i + x_{ij} \beta \nonumber \\
\alpha_i &\sim& N(m, 1/\phi) \nonumber \\
m &\sim& N(0, \kappa^2 / \phi) \nonumber \\
\phi &\sim& Ga(1,1) \nonumber \\
\beta &\sim& N(0, 100 I) \, .\label{eqn:mixed-model}
\end{eqnarray}

An extra step is easily added to the \Polya-Gamma Gibbs sampler to estimate
$(\alpha, \beta, m)$ and $\phi$.  We use the following three data sets to benchmark the \Polya-Gamma method.

\begin{description}

\item[Synthetic:] A synthetically generated dataset with 5 groups, 100 observations within each
group, and a single fixed effect.

\item[Polls:] Voting data from a Presidential campaign \citep{gelman-hill-2006}.
  The response indicates a vote for or against former President George W. Bush.
  There are 49 groups corresponding to states.  Some states have very few
  observations, requiring a model that shrinks coefficients towards a global
  mean to get reasonable estimates.  A single fixed effect for the race of the
  respondent is included, although it would be trivial to include other
  covariates.  Entries with missing data were deleted to yield a total of 2015
  observations.

\item[Xerop:]

The Xerop data set from the \texttt{epicalc} R package \citep{epicalc-2012}.
Indonesian children were observed to examine the causes of respiratory
infections; of specific interest is whether vitamin A deficiencies cause such
illness.  Multiple observations of each individual were made.  The data is
grouped by individual id yielding a total of 275 random intercepts.  A total of
5 fixed effects are included in the model---age, sex, height, stunted growth,
and season---corresponding to an 8 dimensional regression coefficient after
expanding the season covariate using dummy variables.

\end{description}

Table \ref{tab:mm-examples} summarizes the results, which suggest that the \Polya-Gamma method is a sensible default choice for fitting nonlinear mixed-effect models.

While an independence Metropolis sampler usually works well for binary logistic
regression, it does not work well for the mixed models we consider.  For instance, in the polls data set, at least two heuristics that suggest the Laplace approximation will be a poor proposal.
First, the posterior mode does not coincide with the posterior mean.  Second,
the Hessian at the mode is nearly singular.  Its smallest eigenvalue, in
absolute terms, corresponds to an eigenvector that points predominantly in the
direction of $\phi$.  Thus, there is a great deal of uncertainty in the
posterior mode of $\phi$.  If we iteratively solve for the MLE by starting at
the posterior mean, or if we start at the posterior mode for all the coordinates
except $\phi$, which we initialize at the posterior mean of $\phi$, then we
arrive at the same end point.  This suggests that the behavior we observe is not due to a poor
choice of initial value or a poor stopping rule.

The first image in Figure \ref{fig:phi-mm} shows that the difference between the
posterior mode and posterior mean is, by far, greatest in the $\phi$ coordinate.
The second image in Figure \ref{fig:phi-mm} provides one example of the lack of
curvature in $\phi$ at the mode.  If one plots $\phi$ against the other
coordinates, then one sees a similar, though often less extreme, picture.  In
general, large values of $\phi$ are found at the tip of an isosceles triangular
whose base runs parallel to the coordinate that is not $\phi$.  While the upper
tip of the triangle may posses the most likely posterior values, the rest of the
posterior does not fall away quick enough to make that a likely posterior random
variate.

\begin{table}
\tablesize
\centering
\begin{tabular}{l r r r r r r r r } 
\multicolumn{9}{c}{Synthetic: $N=500$, $P_a=5$, $P_b=1$, samp=10,000, burn=2,000, thin=1} \\
\midrule
Method   &     time &    ARate &  ESS.min &  ESS.med &  ESS.max &  ESR.min &  ESR.med &  ESR.max \\ 
PG       &     7.29 &     1.00 &  4289.29 &  6975.73 &  9651.69 &   588.55 &   957.18 &  1324.31 \\ 
Ind-Met. &     3.96 &     0.70 &  1904.71 &  3675.02 &  4043.42 &   482.54 &   928.65 &  1022.38 \\
\\
\multicolumn{9}{c}{Polls: $N=2015$, $P_a=49$, $P_b=1$, samp=100,000, burn=20,000, thin=10} \\
\midrule
Method   &     time &    ARate &  ESS.min &  ESS.med &  ESS.max &  ESR.min &  ESR.med &  ESR.max \\ 
PG       &    31.94 &     1.00 &  5948.62 &  9194.42 &  9925.73 &   186.25 &   287.86 &   310.75 \\ 
Ind-Met. &   146.76 &  0.00674 &    31.36 &    52.81 &    86.54 &     0.21 &     0.36 &     0.59 \\


\\
\multicolumn{9}{c}{Xerop: $N=1200$, $P_a=275$, $P_b=8$, samp=100,000, burn=20,000, thin=10} \\
\midrule
Method   &     time &     ARate &  ESS.min &  ESS.med &  ESS.max &  ESR.min &  ESR.med &  ESR.max \\ 
PG       &   174.38 &      1.00 &   850.34 &  3038.76 &  4438.99 &     4.88 &    17.43 &    25.46 \\ 
Ind-Met. &   457.86 & 0.00002.5 &     1.85 &     3.21 &    12.32 &     0.00 &     0.01 &     0.03
\end{tabular}

\caption{\label{tab:mm-examples} A set of three benchmarks for binary logistic mixed models.  $N$
  denotes the number of samples, $P_a$ denotes the number of groups, and $P_b$
  denotes the dimension of the fixed effects coefficient.  The random effects
  are limited to group dependent intercepts.  Notice that the second and third
  benchmarks are thinned every 10 samples to produce a total of 10,000 posterior
  draws.  Even after thinning, the effective sample size for each is low compared
  to the PG method.  The effective samples sizes are taken for the collection
  $(\alpha, \beta, m)$ and do not include $\phi$.}

\end{table}

\begin{figure}
\label{fig:phi-mm}
\centering
\includegraphics[scale=0.4]{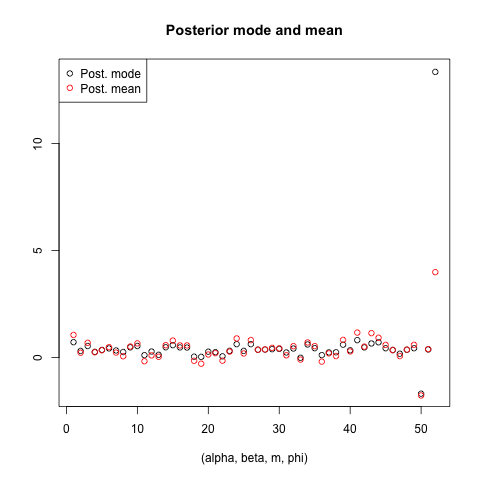}
\includegraphics[scale=0.4]{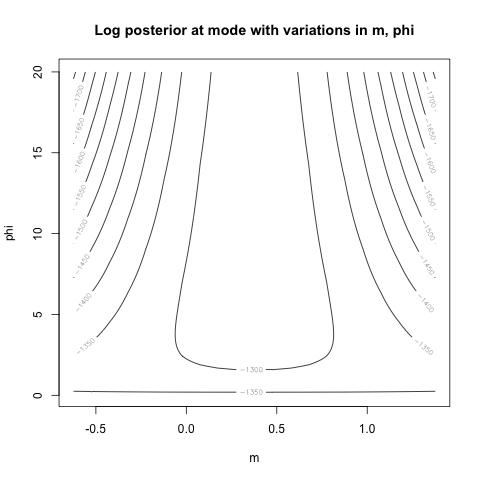}

\includegraphics[scale=0.4]{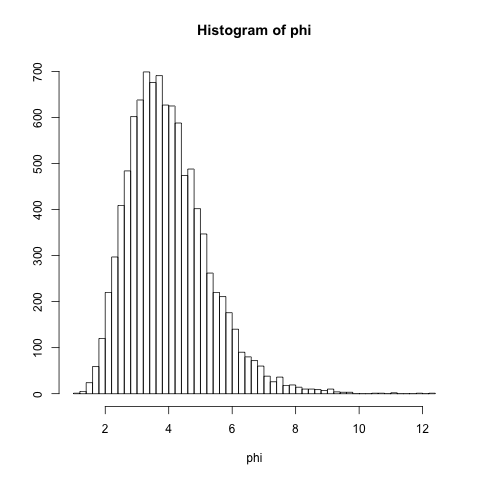}
\includegraphics[scale=0.4]{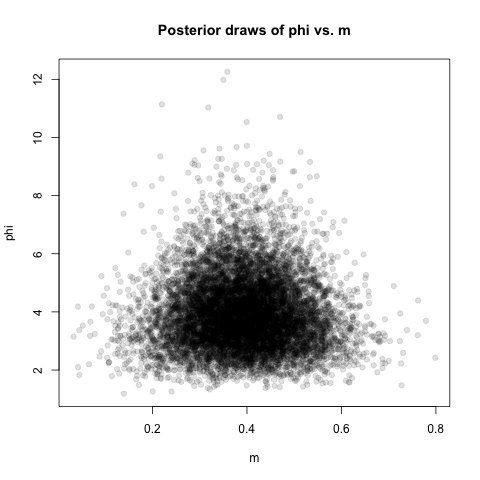}

\caption{Proceeding from left to right and top to bottom.  Upper left: the
  posterior mode and the posterior mean of $(\alpha, \beta, m, \phi)$.  The mode
  and mean are most different in $\phi$.  Upper right: the level sets of $(\phi,
  m)$ of the log posterior when the other coordinates are evaluated at the
  posterior mode.  The log posterior is very flat when moving along $\phi$.
  Bottom left: the marginal posterior distribution of $\phi$.  When
  marginalizing, one finds that few large values of $\phi$ are likely.  Bottom
  right: a scatter plot of posterior samples for $(\phi, m)$. Again, one sees
  that upon marginalizing out the other coordinates the posterior mass is
  concentrated at relatively small values of $\phi$ compared to its value at the
  posterior mode.}
\end{figure}

\section{Benchmarks: negative-binomial models}

\begin{table}
\tablesize
\centering
\begin{tabular}{l r r r r r r r r } 
\multicolumn{9}{c}{Fewer counts: $\alpha = 2$, $\bar y = 8.11$, $\sum y_i = 3244$, $N=400$} \\
\midrule
Method  &     time &    ARate &  ESS.min &  ESS.med &  ESS.max &  ESR.min &  ESR.med &  ESR.max \\ 
    PG  &    26.84 &     1.00 &  7269.13 &  7646.16 &  8533.51 &   270.81 &   284.85 &   317.91 \\ 
    FS  &     8.10 &     1.00 &   697.38 &   719.36 &   759.13 &    86.10 &    88.80 &    93.70 \\ 
   RAM  &    10.17 &    30.08 &   737.95 &   748.51 &   758.57 &    72.59 &    73.62 &    74.61 \\
\\
\multicolumn{9}{c}{More counts: $\alpha = 3$, $\bar y = 23.98$, $\sum y_i = 9593$, $N=400$} \\
\midrule
Method  &     time &    ARate &  ESS.min &  ESS.med &  ESS.max &  ESR.min &  ESR.med &  ESR.max \\ 
    PG  &    58.99 &     1.00 &  3088.04 &  3589.67 &  4377.21 &    52.35 &    60.85 &    74.20 \\ 
    FS  &     8.21 &     1.00 &   901.50 &   915.39 &   935.06 &   109.73 &   111.45 &   113.84 \\ 
   RAM  &     8.69 &    30.33 &   757.91 &   763.81 &   771.73 &    87.25 &    87.93 &    88.84
 \end{tabular}
 \caption{\label{tab:nb-synth}
   Negative binomial regression.  PG is the \Polya-Gamma Gibbs sampler.  FS follows
   \cite{fruhwirth-schnatter-etal-2009}.  RAM is the random walk
   Metropolis-Hastings sampler from the \texttt{bayesm} package \citep{bayesm-2012}. 
   $\alpha$ is the true intercept and $y_i$ is the $i$th response.  Each model
   has three continuous predictors.
 }
\end{table}

\begin{table}
\tablesize
\centering
\begin{tabular}{l r r r r r r r r } 
\multicolumn{9}{c}{Gaussian process 1: $\bar y = 35.7$, $\sum y_i = 9137$, $N=256$, $\ell=0.1$, nugget=$0.0$} \\
\midrule
Method  &     time &    ARate &  ESS.min &  ESS.med &  ESS.max &  ESR.min &  ESR.med &  ESR.max \\ 
    PG  &   101.89 &     1.00 &   790.55 &  6308.65 &  9798.04 &     7.76 &    61.92 &    96.19 \\ 
    FS  &    53.17 &     1.00 &   481.36 &  1296.27 &  2257.27 &     9.05 &    24.38 &    42.45 \\
\\
\multicolumn{9}{c}{Gaussian process 2: $\bar y = 22.7$, $\sum y_i = 22732$, $N=1000$, $\ell=0.1$, nugget=$0.0001$} \\
\midrule
Method  &     time &    ARate &  ESS.min &  ESS.med &  ESS.max &  ESR.min &  ESR.med &  ESR.max \\ 
    PG  &  2021.78 &     1.00 &  1966.77 &  6386.43 &  9862.54 &     0.97 &     3.16 &     4.88 \\ 
    FS  &  1867.05 &     1.00 &   270.13 &  1156.52 &  1761.70 &     0.14 &     0.62 &     0.94

 \end{tabular}
 \caption{\label{tab:nb-gp-ex01}
 Binomial spatial models.  PG is the \Polya-Gamma Gibbs sampler.  FS follows
   \cite{fruhwirth-schnatter-etal-2009}. $N$ is the
   total number of observations and $y_i$ denotes the $i$th observation.}
\end{table}

We simulated two synthetic data sets with $N=400$ data points using the model
$$
y_i \sim NB(\textmd{mean}=\mu_i, d) \; , \quad \log \mu_i = \alpha + x_i \beta
$$
where $\beta \in \mathbb{R}^3$.  Both data sets are included as supplements.  The parameter $d$ is estimated using a
random-walk Metropolis-Hastings step over the integers.  (Neither the \Polya-Gamma method nor the R package by \cite{binomlogit-2012} are set up to work efficiently with non-integer values of this parameter.)  The model with fewer counts corresponds
to $\alpha = 2$, while the model with more counts corresponds to $\alpha = 3$.  This produced a sample mean of roughly 8 in the former case and 24 in the latter.

Table \ref{tab:nb-synth} shows the results for both simulated data sets.  Notice that the \Polya-Gamma method has superior effective sample size in both cases, but a lower effective sampling rate in the second case.  This is caused by the bottleneck of summing $n$ copies of a $\PG(1,z)$ variable to draw a $\PG(n,z)$ variable.  As mentioned in the main manuscript, it is an open challenge to create an efficient \Polya-Gamma sampler for
arbitrary $n$, which would make it the best choice in both cases.

One reaches a different conclusion when working with more complicated models that devote proportionally less time to sampling the auxiliary variables.  Specifically, consider the model
$$
y_i \sim NB(\text{mean} = \mu(x_i), d) \; , \quad
\log \mu  \sim GP(0, K) \, ,
$$
where $K$ is the square exponential covariance kernel,
\[
K(x_1,x_2) = \kappa + \exp \Big( \frac{\|x_1 - x_2\|^2}{2 \ell^2} \Big),
\]
with characteristic length scale $\ell$ and nugget $\kappa$.  Using either the
\Polya-Gamma or \cite{fruhwirth-schnatter-etal-2009} data augmentation
techniques, one arrives at a complete conditional for $\upsilon = \log \mu$ that
is equivalent to the posterior $(\upsilon | z)$ derived using pseudo-data
$\{z_i\}$ generated by
\[
z_i = \upsilon(x_i) + \ep_i, \; \ep_i \sim N(0, V_i)
\]
where $V_i$ is a function of the $i$th auxiliary variable.  Since the prior for
$\upsilon$ is a Gaussian process one may use conjugate formulas to sample the
complete conditional of $\upsilon$.  But producing a random variate from this
distribution is expensive as one must calculate the Cholesky decomposition of a
relatively large matrix at each iteration.  Consequently, the relative time
spent sampling the auxiliary variables in each model decreases, making the
\Polya-Gamma method competitive, and sometimes better, than the method of
Fr\"{u}hwirth-Schnatter et al.  We provide two such examples in Table
(\ref{tab:nb-gp-ex01}).  In the first synthetic data set, 256 equally spaced
points were used to generate a draw $\upsilon(x_i)$ and $y_i$ for $i=1, \ldots,
256$ where $\upsilon \sim GP(0, K)$ and $K$ has length scale $\ell = 0.1$ and a
nugget $\kappa=0.0$.  The average count value of the synthetic data set is $\bar
y = 35.7$, yielding 9137 total counts, which is roughly the same amount as in
the larger negative binomial example discussed earlier.  Now, however, because
proportionally more time is spent sampling the main parameter, and because the
\Polya-Gamma method wastes fewer of these expensive draws, it is more efficient.
In the second synthetic data set, 1000 randomly selected points were chosen to
generate a draw from $\upsilon(x_i)$ and $y_i$ with $\upsilon \sim GP(0, K)$
where $K$ has length scale $\ell = 0.1$ and a nugget $\kappa=0.0001$.  The
average count value is $\bar y = 22.72$, yielding 22,720 total counts.  The
larger problem shows an even greater improvement in performance over the method
of Fr\"{u}hwirth-Schnatter et al.

\section{Extensions}

\subsection{2 $\times$ 2 $\times$ $N$ tables}

Consider a simple example of a binary-response clinical trial conducted in each of $N$ different centers.  Let $n_{ij}$ be the number of patients assigned to treatment regime $j$ in center $i$; and let $Y = \{y_{ij} \}$ be the corresponding number of successes for $i=1, \ldots, N$.  Table 1 presents a data set along these lines, from \citet{skene:wakefield:1990}.  These data arise from a multi-center trial comparing the efficacy of two different topical cream preparations, labeled the treatment and the control.  

Let $p_{ij}$ denote the underlying success probability in center $i$ for treatment $j$, and $\psi_{ij}$ the corresponding log-odds.  If $\psi_i = (\psi_{i1}, \psi_{i2})^T$ is assigned a bivariate normal prior $\psi_i \sim \N(\mu, \Sigma)$ then the posterior for $\Psi = \{\psi_{ij}\}$ is
$$
p(\Psi \mid Y) \propto \prod_{i=1}^N \left\{ \frac{  e^{ y_{i1} \psi_{i1} } }{ ( 1 + e^{\psi_{i1}} )^{n_{i1}} }
 \frac{  e^{ y_{i2} \psi_{i2} } }{ ( 1 + e^{\psi_{i2}} )^{n_{i2}} } \ p( \psi_{i1} , \psi_{i2} \mid \mu, \Sigma) \right\} \, .
$$

We apply Theorem 1 from the main paper to each term in the posterior, thereby introducing augmentation variables $\Omega_i = \mbox{diag}(\omega_{i1}, \omega_{i2})$ for each center.  This yields, after some algebra, a simple Gibbs sampler that iterates between two sets of conditional distributions:
\begin{eqnarray}
(\psi_i \mid Y, \Omega_i, \mu, \Sigma) &\sim& \N(m_i, V_{\Omega_i}) \label{eqn:pgmixture3} \\
(\omega_{ij} \mid \psi_{ij}) &\sim& \PG\left( n_{ij}, \psi_{ij} \right) \nonumber \, ,
\end{eqnarray}
where
\begin{eqnarray*}
V_{\Omega_i}^{-1} &=& \Omega_i + \Sigma^{-1} \\
m_i &=& V_{\Omega_i} (\kappa_i + \Sigma^{-1} \mu) \\
\kappa_i &=& (y_{i1} - n_{i1}/2, y_{i2} - n_{i2}/2)^T \, .
\end{eqnarray*}

Figure \ref{fig:skenewakedata} shows the results of applying this Gibbs sampler to the data from \citet{skene:wakefield:1990}.

\begin{table}
\caption{\label{tab:skenewakeexample} Data from a multi-center, binary-response study on topical cream effectiveness \citep{skene:wakefield:1990}.}
\vspace{1pc}
\begin{center}
\begin{footnotesize}
\begin{tabular}{c  rrrr}
 & \multicolumn{2}{c}{Treatment} & \multicolumn{2}{c}{Control} \\
Center & Success & Total & Success & Total \\ \\
1 & 11 & 36 & 10 & 37\\
2 & 16 & 20 & 22 & 32\\
3 & 14 & 19 & 7 & 19\\
4 & 2 & 16 & 1 & 17\\
5 & 6 & 17 & 0 & 12\\
6 & 1 & 11 & 0 & 10\\
7 & 1 & 5 & 1 & 9\\
8 & 4 & 6 & 6 & 7\\
\end{tabular}
\end{footnotesize}
\end{center}
\end{table}

\begin{figure}
\begin{center}
\includegraphics[width=4.0in]{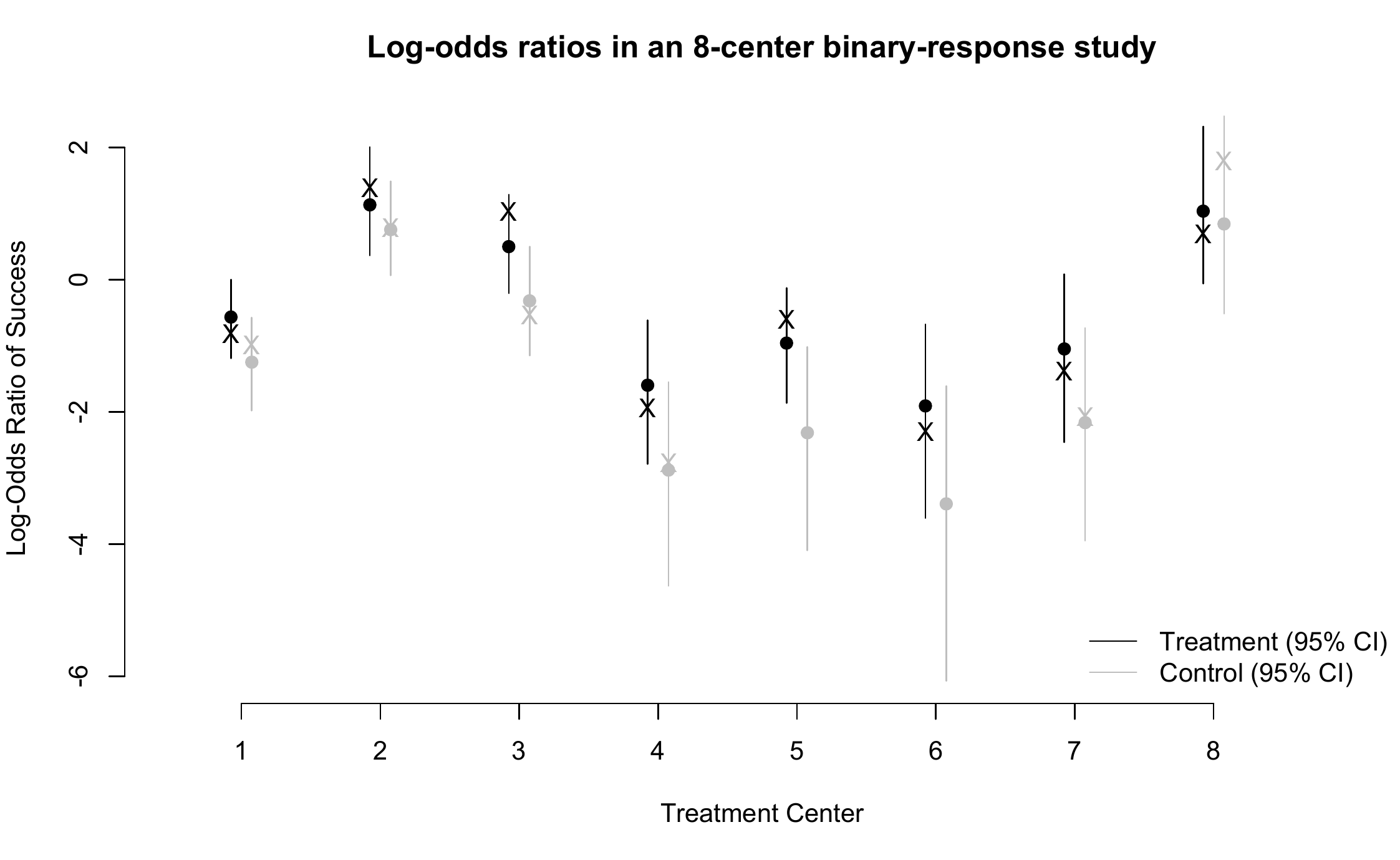}
\caption{\label{fig:skenewakedata} Posterior distributions for the log-odds ratio for each of the 8 centers in the topical-cream study from \citet{skene:wakefield:1990}.  The vertical lines are central $95\%$ posterior credible intervals; the dots are the posterior means; and the X's are the maximum-likelihood estimates of the log-odds ratios, with no shrinkage among the treatment centers.  Note that the maximum-likelihood estimate is $\psi_{i2} = -\infty$ for the control group in centers 5 and 6, as no successes were observed.}
\end{center}
\end{figure}

In this analysis, we used a normal-Wishart prior for $(\mu, \Sigma^{-1})$.   Hyperparameters were chosen to match Table II from \citet{skene:wakefield:1990}, who parameterize the model in terms of the prior expected values for $\rho$, $\sigma^2_{{\psi_1}}$, and $\sigma^2_{{\psi_2}}$, where
$$
\Sigma = \left(
\begin{array}{cc}
\sigma^2_{{\psi_1}} & \rho \\
\rho & \sigma^2_{{\psi_2}}
\end{array}
\right) \, .
$$
To match their choices, we use the following identity that codifies a relationship between the hyperparameters $B$ and $d$, and the prior moments for marginal variances and the correlation coefficient.  If $\Sigma \sim \mathcal{IW}(d, B)$, then
\begin{align*}
B = (d-3) \left[ \begin{array}{cc} \E(\sigma_{\psi_2}^2)+\E(\sigma_{\psi_1}^2)+2\E(\rho) \sqrt{\E(\sigma_{\psi_2}^2)\E(\sigma_{\psi_1}^2)} & \E(\sigma_{\psi_2}^2) + \E(\rho) \sqrt{\E(\sigma_{\psi_2}^2)\E(\sigma_{\psi_1}^2)} \\ \E(\sigma_{\psi_2}^2) + \E(\rho) \sqrt{\E(\sigma_{\psi_2}^2)\E(\sigma_{\psi_1}^2)} & \E(\sigma_{\psi_2}^2) \end{array} \right] \, .
\end{align*}
In this way we are able to map from pre-specified moments to hyperparameters, ending up with $d = 4$ and
$$
B = \left(
\begin{array}{cc}
0.754 & 0.857 \\
0.857 & 1.480
\end{array}
\right) \, .
$$

\subsection{Higher-order tables}

Now consider a multi-center, multinomial response study with more than two treatment arms.  This can be modeled using hierarchy of $N$ different two-way tables, each having the same $J$ treatment regimes and $K$ possible outcomes.  The data D consist of triply indexed outcomes $y_{ijk}$, each indicating the number of observations in center $i$ and treatment $j$ with outcome $k$.  We let $n_{ij} = \sum_k y_{ij}$ indicate the number of subjects assigned to have treatment $j$ at center $i$.

Let $P = \{p_{ijk}\}$ denote the set of probabilities that a subject in center $i$ with treatment $j$ experiences outcome $k$, such that $\sum_k p_{ijk} = 1$ for all $i, j$.  Given these probabilities, the full likelihood is
$$
L(P) = \prod_{i=1}^N \prod_{j=1}^J \prod_{k=1}^K p_{ijk}^{y_{ijk}} \, .
$$

Following \citet{leonard:1975}, we can model these probabilities using a logistic transformation.  Let
$$
p_{ijk} = \frac{\exp(\psi_{ijk})}{\sum_{l=1}^K \exp(\psi_{ijl}) } \, .
$$
Many common prior structures will maintain conditional conjugacy using the Polya-Gamma framework outlined thus far.  For example, we may assume an exchangeable matrix-normal prior at the level of treatment centers:
$$
\psi_i \sim \N(M, \Sigma_R, \Sigma_C) \, ,
$$
where $\psi_{i}$ is the matrix whose $(j,k)$ entry is $\psi_{ijk}$; $M$ is the mean matrix; and $\Sigma_R$ and $\Sigma_C$ are row- and column-specific covariance matrices, respectively. See \citet{dawid81} for further details on matrix-normal theory.  Note that, for identifiability, we set $\psi_{ijK} = 0$, implying that $\Sigma_C$ is of dimension $K-1$.

This leads to a posterior of the form
$$
p(\Psi \mid D) = \prod_{i=1}^N \left[ p(\psi_i) \cdot \prod_{j=1}^J \prod_{k=1}^K  \left( \frac{\exp(\psi_{ijk})}{\sum_{l=1}^K \exp(\psi_{ijl}) } \right)^{y_{ijk}} \right] \, ,
$$
suppressing any dependence on $(M, \Sigma_R, \Sigma_C)$ for notational ease.

To show that this fits within the Polya-Gamma framework, we use a similar approach to \citet{holmes:held:2006}, rewriting each probability as
\begin{eqnarray*}
p_{ijk} &=& \frac{\exp(\psi_{ijk})}{\sum_{l \neq k} \exp(\psi_{ijl})  + \exp(\psi_{ijk})} \\
&=&  \frac{e^{\psi_{ijk}-c_{ijk}}} {1  + e^{\psi_{ijk}-c_{ijk}}} \, ,
\end{eqnarray*}
where $c_{ijk} = \log \{ \sum_{l \neq k} \exp(\psi_{ijl}) \}$ is implicitly a function of the other $\psi_{ijl}$'s for $l \neq k$.

We now fix values of $i$ and $k$ and examine the conditional posterior distribution for $\psi_{i \cdot k} = (\psi_{i1k}, \ldots, \psi_{iJk})'$, given $\psi_{i \cdot l}$ for $l \neq k$:
\begin{eqnarray*}
p( \psi_{i \cdot k} \mid D, \psi_{i \cdot (-k)}) &\propto& p(\psi_{i \cdot k} \mid \psi_{i \cdot (-k)}) \cdot  \prod_{j=1}^J  \left( \frac{e^{\psi_{ijk}-c_{ijk}}} {1  + e^{\psi_{ijk}-c_{ijk}}} \right)^{y_{ijk}} \left( \frac{1}{1  + e^{\psi_{ijk}-c_{ijk}}} \right)^{n_{ij} - y_{ijk}} \\
&=&  p(\psi_{i \cdot k} \mid \psi_{i \cdot (-k)}) \cdot  \prod_{j=1}^J  \frac{e^{y_{ijk}(\psi_{ijk}-c_{ijk})}} {(1  + e^{\psi_{ijk}-c_{ijk}})^{n_{ij}}}
\end{eqnarray*}

This is simply a multivariate version of the same bivariate form in that arises in a $2 \times 2$ table.  Appealing to the theory of Polya-Gamma random variables outlined above, we may express this as:
\begin{eqnarray*}
p( \psi_{i \cdot k} \mid D, \psi_{i \cdot (-k)}) &\propto& p(\psi_{i \cdot k} \mid \psi_{i \cdot (-k)}) \cdot  \prod_{j=1}^J  
\frac{  e^ { \kappa_{ijk} [\psi_{ijk} - c_{ijk}] } } { \cosh^{n_{ij}} ([\psi_{ijk} - c_{ijk}]/2)} \\
&=&  p(\psi_{i \cdot k} \mid \psi_{i \cdot (-k)}) \cdot \prod_{j=1}^J \left[ e^{ \kappa_{ijk} [\psi_{ijk} - c_{ijk}]} \cdot 
\E \left\{ e^{-\omega_{ijk} [\psi_{ijk} - c_{ijk}]^2/2 } \right\} \right] \, ,
\end{eqnarray*}
where $\omega_{ijk} \sim \PG(n_{ij}, 0)$, $j=1,\ldots,J$; and $\kappa_{ijk} = y_{ijk} - n_{ij}/2$.  Given $\{ \omega_{i j k} \}$ for $j=1, \ldots, J$, all of these terms will combine in a single normal kernel, whose mean and covariance structure will depend heavily upon the particular choices of hyperparameters in the matrix-normal prior for $\psi_i$.  Each $\omega_{ijk}$ term can be updated as
$$
(\omega_{ijk} \mid \psi_{ijk}) \sim \PG(n_{ij}, \psi_{ijk} - c_{ijk}) \, ,
$$
leading to a simple MCMC that loops over centers and responses, drawing each vector of parameters $\psi_{i \cdot k}$ (that is, for all treatments at once) conditional on the other $\psi_{i \cdot (-k)}$'s.

\subsection{Multinomial logistic regression}

One may extend the \Polya-Gamma method used for binary logistic regression to
multinomial logistic regression.  Consider the multinomial sample $y_i =
\{y_{ij}\}_{j=1}^J$ that records the number of responses in each category $j=1,
\ldots, J$ and the total number of responses $n_i$.  The logistic link function
for polychotomous regression stipulates that the probability of randomly drawing
a single response from the $j$th category in the $i$th sample is
\[
p_{ij} = \frac{\exp {\psi_{ij}}}{\sum_{i=1}^J \exp {\psi_{ik}}}
\]
where the log odds $\psi_{ij}$ is modeled by $x_i^T \bbeta_j$ and $\bbeta_J$ has
been constrained to be zero for purposes of identification.  Following Holmes
and Held (2006) the likelihood for $\bbeta_j$ conditional upon $\bbeta_{-j}$,
the matrix with column vector $\bbeta_j$ removed, is
\[
\ell(\bbeta_j | \bbeta_{-j}, y) = \prod_{i=1}^N 
\left( \frac{ e^{\eta_{ij}} }{ 1+e^{\eta_{ij}} } \right)^{y_{ij}} 
\left( \frac{ 1 }{ 1+e^{\eta_{ij}} } \right)^{n_i - y_{ij}} 
\]
where 
\[
\eta_{ij} = x_i^T \bbeta_j - C_{ij} \textmd{ with } C_{ij} = \log \sum_{k \neq j}
\exp{x_i^T \bbeta_k},
\]
which looks like the binary logistic likelihood previously discussed.
Incorporating the \Polya-Gamma auxiliary variable, the likelihood becomes
\[
\prod_{i=1}^N e^{\kappa_{ij} \eta_{ij}} e^{-\frac{\eta_{ij}^2}{2}} \omega_{ij}
    PG(\omega_{ij} | n_i, 0)
\]
where $\kappa_{ij} = (y_{ij} - n_i / 2)$.  Employing the conditionally conjugate
prior $\bbeta_j \sim N(m_{0j}, V_{0j})$ yields a two-part update:
\begin{align*}
(\bbeta_j \mid \Omega_j) & \sim N(m_{j}, V_{j}) \\
(\omega_{ij} \mid \bbeta_j) & \sim PG(n_i, \eta_{ij}) \textmd{ for } i = 1, \cdots, N,
\end{align*}
where
\begin{align*}
V_j^{-1} & = X' \Omega_j X + V_{0j}^{-1}, \\
m_j & = V_j \left( X' (\bm{\kappa}_j - \Omega_j c_j) + V_{0j}^{-1} m_{0j} \right),
\end{align*}
$c_j$ is the $j$th column of $C$, and $\Omega_j =
\diag(\{\omega_{ij}\}_{i=1}^N)$.  One may sample the posterior distribution of
$(\bbeta \mid y)$ via Gibbs sampling by iterating over the above steps for $j=1,
\ldots, J-1$.

The \Polya-Gamma method generates samples from the joint posterior distribution without appealing to analytic approximations to the posterior.  This offers an important advantage when the number of observations is not significantly
larger than the number of parameters.  

To see this, consider sampling the joint posterior for $\bbeta$ using a Metropolis-Hastings algorithm with an independence proposal.  The likelihood in $\bbeta$ is approximately normal, centered at the posterior mode $m$, and with variance $V$ equal to the inverse of the Hessian matrix evaluated at the mode.  (Both of these may be found using standard numerical routines.)  Thus a natural proposal for $(\vect(\beta^{(t)}) \mid y)$ is $\vect(b) \sim N(m, aV)$ for some $a \approx 1$.  When data are plentiful, this method is both simple and highly efficient, and is implemented in many standard software packages \citep[e.g.][]{mcmcpack:2011}.

But when $\vect(\beta)$ is high-dimensional relative to the number of
observations the Hessian matrix $H$ may be ill-conditioned, making it impossible
or impractical to generate normal proposals.  Multinomial logistic regression
succumbs to this problem more quickly than binary logistic regression, as the
number of parameters scales like the product of the number of categories and the
number of predictors.

To illustrate this phenomenon, we consider glass-identification data from
\cite{glassdataset}.  This data set has $J=6$ categories of glass and nine
predictors describing the chemical and optical properties of the glass that one
may measure in a forensics lab and use in a criminal investigation.  This
generates up to $50 = 10 \times 5$ parameters, including the intercepts and the
constraint that $\beta_J = 0$.  These must be estimated using $n = 214$
observations.  In this case, the Hessian $H$ at the posterior mode is poorly
conditioned when employing a vague prior, incapacitating the independent
Metropolis-Hastings algorithm.  Numerical experiments confirm that even when a
vague prior is strong enough to produce a numerically invertible Hessian,
rejection rates are prohibitively high.  In contrast, the multinomial
\Polya-Gamma method still produces reasonable posterior distributions in a fully
automatic fashion, even with a weakly informative normal prior for each
$\beta_j$.  Table \ref{tab:mult-logit}, which shows the in-sample performance of
the multinomial logit model, demonstrates the problem with the joint proposal
distribution: category 6 is perfectly separable into cases and non-cases, even
though the other categories are not.  This is a well-known problem with
maximum-likelihood estimation of logistic models.  The same problem also
forecloses the option of posterior sampling using methods that require a unique
MLE to exist.

\begin{table}
\centering
\begin{tabular}{l c c c c c c}
Class   &  1 &  2 &  3 &  5 & 6 &  7 \\
\midrule
Total   & 70 & 76 & 17 & 13 & 9 & 29 \\
Correct & 50 & 55 & 0  &  9 & 9 & 27
\end{tabular}
\caption{``Correct'' refers to the number of glass fragments for each
category that were correctly identified by the Bayesian multinomial logit model.  The glass identification dataset includes a type of glass, class 4,
for which there are no observations.}
\label{tab:mult-logit}
\end{table}

\end{document}